\documentclass[11pt]{article}
\usepackage{amsmath,amssymb,amsthm,algpseudocode,algorithm,float,tkz-berge,caption,mathtools,thmtools,thm-restate,scalerel,stackengine}
  
\stackMath
\newcommand\reallywidehat[1]{%
\savestack{\tmpbox}{\stretchto{%
  \scaleto{%
      \scalerel*[\widthof{\ensuremath{#1}}]{\kern-.6pt\bigwedge\kern-.6pt}%
          {\rule[-\textheight/2]{1ex}{\textheight}}%WIDTH-LIMITED BIG WEDGE
            }{\textheight}% 
            }{0.5ex}}%
            \stackon[1pt]{#1}{\tmpbox}%
            }
\usepackage[margin=1in]{geometry}
\usepackage{multirow}
\usepackage{subcaption}
\usepackage{afterpage}
\usepackage{bm}
\usepackage[hidelinks]{hyperref}
\usepackage{tikz}
\usepackage{bbm}

\newcommand{\supp}{\text{supp}}

\newcommand{\wt}[1]{\widetilde{#1}}
\newcommand{\Ot}[1]{\widetilde{O}}
\newcommand{\wh}[1]{\widehat{#1}}
\newcommand{\abs}[1]{|#1|}
\newcommand{\eps}{\varepsilon}
\newcommand{\R}{\mathbb{R}}
\newcommand{\E}[2][]{\operatorname*{\mathbb{E}}_{#1 }\left\lbrack #2 \right\rbrack}
\newcommand{\Pb}[2][]{\operatorname*{\mathbb{P}}_{#1 }\left\lbrack #2 \right\rbrack}

\newcommand{\e}{{\varepsilon}}

\DeclareMathOperator*{\argmax}{arg\,max}

\DeclareMathOperator{\promise}{\mathtt{PromiseCounting}}
\DeclareMathOperator{\nopromise}{\mathtt{Counting}}
\DeclareMathOperator{\var}{\text{Var}}

\newtheorem{theorem}{Theorem}
\newtheorem*{theorem*}{Theorem}
\newtheorem{lemma}[theorem]{Lemma}
\newtheorem{definition}[theorem]{Definition}
\newtheorem*{definition*}{Definition}
\newtheorem*{lemma*}{Lemma}
\newtheorem{corollary}[theorem]{Corollary}
\newtheorem*{corollary*}{Corollary}
\newtheorem{claim}[theorem]{Claim}
\newtheorem*{claim*}{Claim}

\title{The Sketching Complexity of Graph and Hypergraph Counting}
\author{John Kallaugher\thanks{This work was done in part while the authors were visiting the Simons Institute for the Theory of Computing.}\\jmgk@cs.utexas.edu\\UT Austin \and Michael Kapralov\thanks{Supported in part by ERC Starting Grant 759471-Sublinear.}\\michael.kapralov@epfl.ch\\EPFL \and Eric Price\footnotemark[1]\\ecprice@cs.utexas.edu\\UT Austin}

{\makeatletter
	\gdef\xxxmark{%
		\expandafter\ifx\csname @mpargs\endcsname\relax % in minipage?
		\expandafter\ifx\csname @captype\endcsname\relax % in figure/caption?
		\marginpar{xxx}% not in a caption or minipage, can use marginpar
		\else
		xxx % notice trailing space
		\fi
		\else
		xxx % notice trailing space-
		\fi}
	\gdef\xxx{\@ifnextchar[\xxx@lab\xxx@nolab}
	\long\gdef\xxx@lab[#1]#2{{\bf [\xxxmark #2 ---{\sc #1}]}}
	\long\gdef\xxx@nolab#1{{\bf [\xxxmark #1]}}
}

\begin{document}

\begin{titlepage}
  
\maketitle

\begin{abstract}
Subgraph counting is a fundamental primitive in graph processing, with applications in social network analysis (e.g.,\ estimating the clustering coefficient of a graph), database processing and other areas. The space complexity of subgraph counting has been studied extensively in the literature, but many natural settings are still not well understood. In this paper we revisit the subgraph (and hypergraph) counting problem in the sketching model, where the algorithm's state as it processes a stream of updates to the graph is a linear function of the stream. This model has recently received a lot of attention in the literature, and has become a standard model for solving dynamic graph streaming problems. 

In this paper we give a tight bound on the sketching complexity of counting the number of occurrences of a small subgraph $H$ in a bounded degree graph $G$ presented as a stream of edge updates. Specifically, we show that the space complexity of the problem is governed by the fractional vertex cover number of the graph $H$. Our subgraph counting algorithm implements a natural vertex sampling approach, with sampling probabilities governed by the vertex cover of $H$. Our main technical contribution lies in a new set of Fourier analytic tools that we develop to analyze multiplayer communication protocols in the simultaneous communication model, allowing us to prove a tight lower bound. We believe that our techniques are likely to find applications in other settings. Besides giving tight bounds for all graphs $H$, both our algorithm and lower bounds extend to the hypergraph setting, albeit with some loss in space complexity.
\end{abstract}

\thispagestyle{empty}
%\pagebreak
\end{titlepage}

%!TEX root = ./main.tex
\section{Introduction}

Triangle counting is one of the most well-studied problems in
streaming graph algorithms.  In the standard ``turnstile'' version of
this problem, one maintains a small-space ``sketch'' of a
graph under a stream of insertions and deletions of edges, and at the
end of the stream outputs a $(1 \pm \eps)$ multiplicative approximation
to the number of triangles $T$ in the graph; unless otherwise
specified, we assume $\eps$ to be a small constant.

Turnstile streaming algorithms are almost invariably constructed as
\emph{linear} sketches, where the sketch maintained is a linear
function of the indicator vector of edges; this makes it easy to
process insertions and deletions.  Linear sketches are also useful in
other settings such as distributed computation, since sketches can be
merged.  There is evidence that any turnstile streaming algorithm can
be efficiently implemented using linear
sketches~\cite{li2014turnstile,ai2016new}, although these results do
not quite apply to graph streams.

For worst-case graphs, counting triangles is impossible in sublinear
space: $\Omega(m)$ space is required to distinguish between a
graph with $0$ triangles and one with $T = \Omega(m)$
triangles~\cite{BOV13}.  However, the hard case is degenerate in that
all the triangles share a common edge.  If at most $\Delta_E$
triangles share any single edge, then this bound becomes $\Omega(m\Delta_E/T)$.
In~\cite{PT12} an algorithm was given that
counts triangles with
\[
  O\left(m \left(\frac{1}{\sqrt{T}} + \frac{\Delta_E}{T}\right)\right)
\]
space, where the $m/\sqrt{T}$ term improves upon a $m/T^{1/3}$ term
in~\cite{TKMF09,TKM11}. In~\cite{KP17} this was shown to be tight for
worst-case graphs, but the hard case is again degenerate: all the
triangles share a common vertex.  If $\Delta_V$ bounds the maximum number
of triangles to share a vertex, this bound becomes $m\sqrt{\Delta_V}/T$.
The algorithm in~\cite{KP17} 
requires
\[
  \Ot{}\left(m \left(\frac{1}{T^{2/3}} + \frac{\sqrt{\Delta_V}}{T} + \frac{\Delta_E}{T}\right)\right)
\]
space.  A natural question is whether this $\frac{m}{T^{2/3}}$ is
necessary.

\subsection{Linear Sketching}
Suppose we have a problem of the following form: we receive a vector $v \in \mathbb{Z}^n$ as a series of updates $(v_i)_{i = 1}^t$, so $v = \sum_{i = 1}^t v_i$, and we want to approximate some function $f(v)$. A \emph{linear sketch} for this problem is a linear transformation $A \in \mathbb{Z}^{n \times d}$ and a post-processing function $g$, so that $g\left(Av\right)$ approximates $f(v)$ (with the exact definition of ``approximates'' depending on the problem). The space complexity of such a sketch is the space needed to store the sketch vector, which is $\Theta(d \log n)$ bits if the maximum size of entries of $v$ is bounded by some $M = poly(n)$ (this holds even if intermediate stages of the stream exceed $M$, as the sketch vector may be stored $\bmod\; M$).

In \cite{li2014turnstile}, it was shown that any turnstile streaming algorithm (an algorithm that can solve a problem of the above form when the updates $v_i$ are allowed to be negative, but that is allowed to maintain arbitrary state) can be converted into a linear sketch with only logarithmic loss in space complexity. In \cite{ai2016new}, this result was extended to \emph{strict} turnstile streaming algorithms, that is algorithms which require that $\sum_{i = 1}^{s}v_i \ge 0$ for each $s \le t$.

These results come with two important caveats. Firstly, they do not necessarily give an $O(d \log n)$-space streaming algorithm, as neither the linear transformation $A$ nor the post-processing function $g$ is known to be calculable in $O(d \log n)$ space. However, our sketching lower bounds will be based on communication complexity arguments that bound the size of the sketch vector, circumventing this issue.

Secondly, \cite{li2014turnstile} assumes that the turnstile algorithm in question works regardless of the value of the partial sums $\sum_{i = 1}^{s} v_i$, while \cite{ai2016new} only requires that these sums be non-negative. Therefore, our lower bounds do not rule out the possibility of a turnstile algorithm that requires every partial sum to form a valid graph (i.e.\ edges can neither be deleted before they arrive, nor can they arrive multiple times before being deleted). However, existing turnstile algorithms do not typically require this property---in particular, any sampling-based algorithm, whether adaptive or non-adaptive, can handle it by the addition of a counter to each stored edge.

\subsection{Our Results}

\paragraph{Lower bound.} We show that any linear
sketching algorithm for triangle counting requires
$\Omega(\frac{m}{T^{2/3}})$ space, even for constant degree graphs.
Such a result is not true for the insertion-only model, where
triangles can be counted in graphs with max degree $d$ in $O(m d^2/T)$ space by subsampling
the edges at rate $d/T$ and storing all subsequent edges that touch
the sampled edges~\cite{JG05}.

Our result generalizes to counting the number of copies of any
constant-size connected subgraph $H$.  Such problems appear, for
example, in estimating the size of database joins when planning
queries~\cite{atserias2008size}.  We show that the linear sketching complexity of
distinguishing between $0$ and $T$ copies of $H$ in constant-degree
graphs is at least
\[
  \Omega\left(\frac{m}{T^{1/\tau}}\right)
\]
where $\tau$ is the fractional vertex cover number of $H$, the minimum value such that 
there exists $f : V(H) \rightarrow \lbrack 0, 1\rbrack$ with $\sum_{v \in V(H)} f(v) \le \tau$
and $f(u) + f(v) \ge 1$ for all $uv \in E(H)$.

\paragraph{Upper bound.}  We also give a matching upper bound: by
subsampling the vertices with probabilities dependent on their weight
in the fractional vertex cover, we give an algorithm that estimates
$T$ using $O(m/T^{1/\tau})$ words of space, as long as the graph has
constant degree.  Additionally, the constant-degree restriction can be
lifted for many graphs: if an optimal fractional vertex cover of $H$
can place nonzero weight on every vertex (as occurs, for example, if
$H$ is a cycle) then the algorithm works for degree up to $T^{1/(2\tau)}$
graphs.

\paragraph{Hypergraph counting.}  Both our upper and lower bounds
extend to counting hypergraphs $H$, but the exponent on $T$ no longer
matches for all hypergraphs.  The upper bound remains
$O(m/T^{1/\tau})$, while the lower bound becomes $O(m/T^{1/\mu})$ for
an exponent $\mu$ that equals the fractional vertex cover number $\tau$ on
many hypergraphs but not all.

All of our results extend to $\eps \ll 1$; the full statements of these results
are given in Theorem \ref{thm:hypergraphs-ub} for the upper bound, Theorem \ref{thm:graphsketching}
for the general lower bound, and Corollary \ref{cor:hypographlb} for the
tight lower bound specific to non-hypergraphs.

\subsection{Sampling and Sketching}
Our upper bounds will take the form of \emph{non-adapative} sampling algorithms. By ``sampling algorithm'' we mean that the only state maintained between updates is a subset of the input edges, and by ``non-adaptive'' we mean that the probability of keeping an edge does not depend on which other edges have been seen so far in the stream\footnote{Note that this does not mean that the edges are sampled \emph{independently} of one another---for instance, if we choose a vertex at random and keep all edges incident on that vertex, the event that we keep the edge $uv$ is independent of whether the edge $vw$ is present in the stream, but it is not independent of the event that that we keep the edge $vw$.}.

Non-adaptive sampling algorithms may be modifed into linear sketching algorithms by the use of $L_0$-sampler sketches. An $L_0$-sampler sketch is a linear sketch which, if $v$ is the frequency vector of the input stream, returns a non-zero co-ordinate of $v$, chosen uniformly at random (more generally, an $L_p$-sampler samples $v_i$ with probability proportional to $|v_i|^p$). A linear sketching algorithm for this problem was first presented in \cite{CMR05}, while \cite{MW10} defined $L_p$ sampling and gave algorithms for all $p \in \lbrack 0, 2\rbrack$. 

In \cite{JST11}, it was shown that, if the set to be sampled from has size $n$ and the sample is required to be within $\delta$ of uniform for some constant $\delta$, the space required is exactly $\Theta(\log^2 n)$. In \cite{KNP+17}, the optimal bound in terms of $n$ and $\delta$ was shown to be $\Theta\left(\min\left(n, \log\left(1/\delta\right)\log^2\left(\frac{n}{\log(1/\delta)}\right)\right)\right)$.

Therefore, as our sketching lower bound and sampling upper bounds match up to polylog factors, it follows that both are themselves tight up to polylog factors.

\subsection{Our Techniques}

The core of our lower bound proof is a new set of Fourier analytic techniques for analyzing multiplayer simultaneous communication protocols. Our approach is inspired by the Fourier analytic analysis of the Boolean Hidden Matching problem, but develops several new ideas that we think are likely to find applications beyond subgraph counting lower bounds. We now proceed to describe the Boolean Hidden Matching problem, the main ideas behind its analysis, and then describe our techniques. \if 0 Our linear sketching lower bound relies on showing a communication
complexity lower bound in the simultaneous message model.  Our proof
of the communication complexity lower bound blends Fourier analysis
with a combinatorial argument in a way that we believe is of
independent interest.
\fi 

The Boolean Hidden Matching problem of Gavinsky et al~\cite{GKKRd07} is a two player one way communication problem where Alice holds a binary string $x\in \{0, 1\}^n$ that she compresses to a message $m$ of $s$ bits and sends to Bob. Bob, besides the message from Alice, gets two pieces of input: a uniformly random matching of size $n/10$ on the set $\{1, 2, \ldots, n\}$, along with a vector of binary labels $w_e\in \{0 ,1\}, e\in M$. In the YES case of the problem the vector $w$ satisfies $w=Mx$, and in the NO case of the problem the vector $w$ satisfies $w=Mx\oplus 1^{|M|}$, where we abuse notation somewhat by letting $M$ denote the edge incidence matrix of the matching $M$ (where each column corresponds to a vertex, and each edge to a row, with ones in the two co-ordinates corresponding to the vertices the edge is incident to). 

The Boolean Hidden Matching problem and the related Boolean Hidden Hypermatching problem of Verbin and Yu~\cite{VY11} have been very influential in streaming lower bounds: streaming problems that have recently been shown to admit reductions from Boolean Hidden (Hyper)Matching include approximating maximum matching size~\cite{EsfandiariHLMO15,AssadiKLY15, AssadiKL17}, approximating MAX-CUT value~\cite{KKS15,kogan-krauthgamer14, KapralovKSV17}, subgraph counting~\cite{VY11, KP17}, and approximating Schatten $p$-norms~\cite{LiW16}, among others. Most recent streaming lower bounds (with the exception of~\cite{KapralovKSV17}) use reductions from Boolean Hidden (Hyper)Matching, without modifying the Fourier analytic techniques involved in the proof. 

In this paper we develop several new Fourier analytic ideas that go beyond the Boolean Hidden Matching problem in several directions:

\paragraph{Analyzing simultaneous multiplayer communication.} In the Boolean Hidden Matching problem Alice is the only player transmitting a message, but our communication problem features simultaneous communication from multiple players to a referee. We show how to use the convolution theorem from Fourier analysis to analyze the effect of combining information sent by multiple players in the one way simultaneous communication model. While the technique of combining information from two players using the convolution theorem was recently used by~\cite{KapralovKSV17} to analyze a three-player game, to the best of our knowledge our work is the first to analyze games with an arbitrary number of players in this manner.

\paragraph{Analyzing a promise version of a communication problem via Fourier analysis.} While in the Boolean Hidden Matching problem Alice's string is sampled from the uniform distribution, in our problem multiple players receive correlated binary strings (conditioned on a linear constraint over the binary field). It turns out that this specific form of conditioning lends itself naturally to a Fourier analytic approach due to the linearity of the constraints imposed on the strings, and analysing such correlated settings gives us tight bounds on the subgraph counting version of our problem.

\paragraph{Sharing $M$ among the players.} In the Boolean Hidden Matching problem, only Bob has the linear function $M$, while Alice must send her message based only on $x$. In our communication problem each (hyper)edge in $H$ corresponds to a player, and every player receives a linear sized set of edges, together with parities of a hidden string $x$ over these edges. Similarly to the Boolean Hidden Matching problem, these parities are either correct (YES case) or flipped simultaneously. A crucial new component, however, is the fact that instead of $M$ being held by the recipient alone, each player holds part of it, and the parts the players hold are correlated. 

Analyzing such correlations is in fact necessary even if one only wants to prove a simple lower bound on the space complexity of `sampling-type' algorithms for triangle counting. We show how to analyze such correlations when $H$ is an arbitrary hypergraph through a purely combinatorial lemma. The weights lemma (Section~\ref{sec:weights}) is primarily concerned with the ability of the players to co-ordinate ``weight'' functions. This can be used to lower bound the space complexity of sampling-based protocols---we apply it to bound the Fourier coefficients of the referee's posterior distribution on the players' inputs when the players send \emph{arbitrary} messages.

\subsection{Related Work}

The past decade has seen a large amount of work on the space complexity of graph problems in the streaming model of computation (see, e.g. the recent survey by McGregor~\cite{McGregor17}). The semi-streaming model of computation, which allows $\tilde O(n)$ space to process a graph on $n$ vertices, has been extensively studied, with space efficient algorithms known for many fundamental graph problems such as spanning trees~\cite{agm}, sparsifiers~\cite{anh-guha, kl11, agm_pods, KapralovLMMS14}, 
matchings~\cite{guha-ahn-1, guha-ahn-3, gkk:streaming-soda12, kap13, GO12, 
HRVZ13, Konrad15, AssadiKLY15}, spanners~\cite{agm_pods, KW14}. Beyond the semi-streaming model, 
it has recently been shown that it is sometimes possible to approximate the {\em cost} of 
the solution to a graph problem in the streaming model even when the amount of space available does not suffice to store the vertex set of the graph 
(e.g.~\cite{KKS14, EsfandiariHLMO15, McGregor15, CormodeJMM17, 0001V18, PengS18}).  The problem of designing lower bounds for graph sketches has received a lot of attention recently due to the success of graph sketching as an approach to solving dynamic graph streaming problems (e.g.,~\cite{LiNW14, AssadiKLY15, AssadiKL17}). Similarly to our approach, such lower bounds normally make use of the simultaneous communication model.

\paragraph{Subgraph counting.}
The streaming subgraph counting problem was introduced in~\cite{BKS02} for the case where $H$ is a triangle. This was followed by alternative algorithms in~\cite{BFLMS06,JG05}. 
The lower bounds in~\cite{BOV13} and~\cite{KP17} were achieved by reductions to one-way communication complexity problems, the indexing problem and the Boolean Hidden Matching problem, respectively. Triangle detection has also been studied as a pure communication problem, for instance \cite{FGO17}, as well as in the adjacency-list~\cite{KMPT10,BFLMS06, MVV16}, multi-pass~\cite{BOV13, CJ14}, and query models~\cite{ELRS15}.

Work on counting non-triangle subgraphs includes \cite{BFLS07}, which presented an algorithm for counting copies of $K_{3,3}$, \cite{BDGL08}, which studied subgraphs of size $3$ and $4$, \cite{MMPS11}, which studied cycles of arbitrary size, and \cite{KMSS12}, which studied arbitrary subgraphs. The problem has also been studied in the query~\cite{JSP15, ANRD15, PSV17} and distributed~\cite{ESBD15,ESBD16} models.

\paragraph{Join size estimation.}  The size of a database join can be
viewed as a ``labeled'' version of hypergraph counting, where each
vertex of $G$ can only match a particular vertex (``attribute'') of
$H$, and each hyperedge of $G$ can only match a particular hyperedge
(``relation'') of $H$.  (Both our upper and lower bounds apply in
this labeled setting.)  In~\cite{atserias2008size} it was shown for a database $G$
with $m$ hyperedges, the size of the join given by a query $H$ can be up to $\Theta(m^\rho)$,
where $\rho$ is the fractional \emph{edge} cover number of $H$.

This result is from a very different regime from ours because it
involves very dense graphs and ours involves sparse ones.  But one
intriguing connection is through the $\Omega(m/T^{2/3})$ lower bound
given in~\cite{KP17} for the restricted class of ``triangle sampling''
algorithms.  Generalizing that proof for arbitrary $H$ would
use~\cite{atserias2008size} to get a lower bound of
$\Omega(m (\frac{1}{T n^{2\rho - \abs{V}}})^{1/\rho})$, as
opposed to our $\Omega(m/T^{1/\tau})$ bound.  These are the same for
some graphs, such as odd cycles, where $\rho = \tau = \abs{V}/2$, and
$\Omega(m/T^{1/\tau})$ is stronger for sparse graphs, but the two
bounds are incomparable in general.  It seems possible that the sample
complexity for dense graphs will depend on $\rho$ in some fashion.

\section{Proof Overview}

\subsection{Lower Bound}

For a fixed graph $H$ with fractional vertex cover $\tau$, we prove an
$\Omega(n/T^{1/\tau})$ lower bound for determining whether a
constant-degree graph on $\Theta(n)$ vertices has $0$ or $\Theta(T)$
copies of $H$.  For illustration, in this section we focus on the case
where $H$ is a triangle.

We consider the three-party simultaneous-message communication problem
illustrated in Figure~\ref{fig:instance}, where each player is
associated with an edge of $H$.  First, we construct a set of
$N = \Theta(n)$ vertices $V_u$ for each vertex $u \in H$.  The player
associated with edge $e = (u, v)$ receives an input consisting of $n$
disjoint edges on $V_u \times V_v$, along with binary labels
associated with each edge.  We are guaranteed that the three players'
inputs collectively contain $T$ triangles, with all the other edges
disjoint. Each set of vertices are randomly permuted so that the players do not
know which of their edges participate in triangles.

We also guarantee that the XOR of the labels
associated with a triangle is the same for every triangle---either
every triangle has an even number of 1s, or an odd number.  The goal
of the players is to send messages to a referee who knows the edges
but not their labels, and for the referee to figure out if every
triangle has an odd number of 1 labels.

\begin{figure}
  \centering
    \begin{subfigure}[t]{0.47\textwidth}
      \centering
      \includegraphics[width=\textwidth]{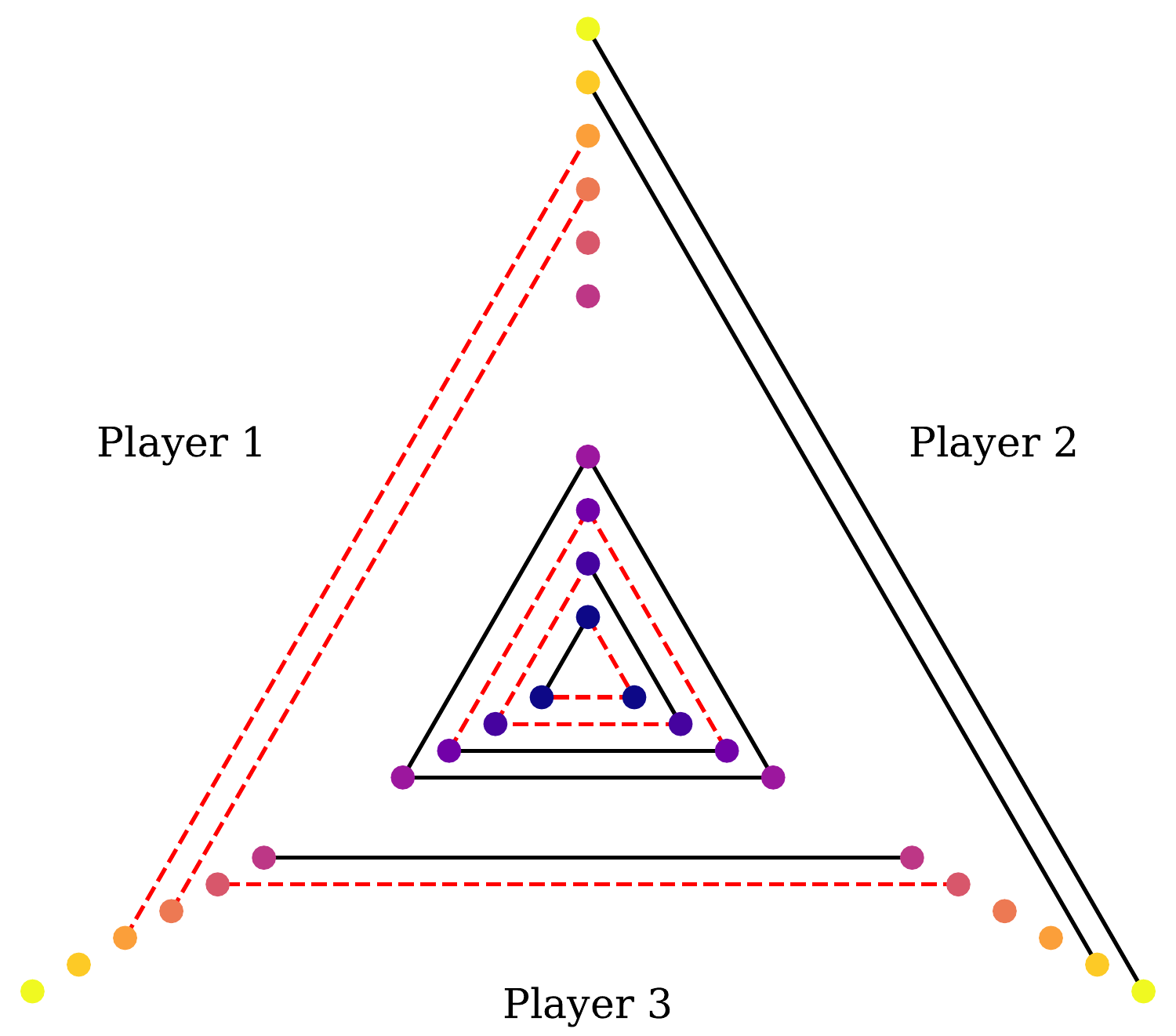}
      \caption{Input for triangle-counting before permutation.  Each
        player $e$ sees $n$ edges with associated binary labels
        (pictured as solid/dashed).  The edges match up into $T$
        triangles (center) and $n-T$ isolated edges (outside).  The
        goal is to determine whether every triangle has an even number
        of solid edges, or an odd number.}
      \label{fig:fig1a}
    \end{subfigure}\hfill
    \begin{subfigure}[t]{0.47\textwidth}
      \centering
      \includegraphics[width=\textwidth]{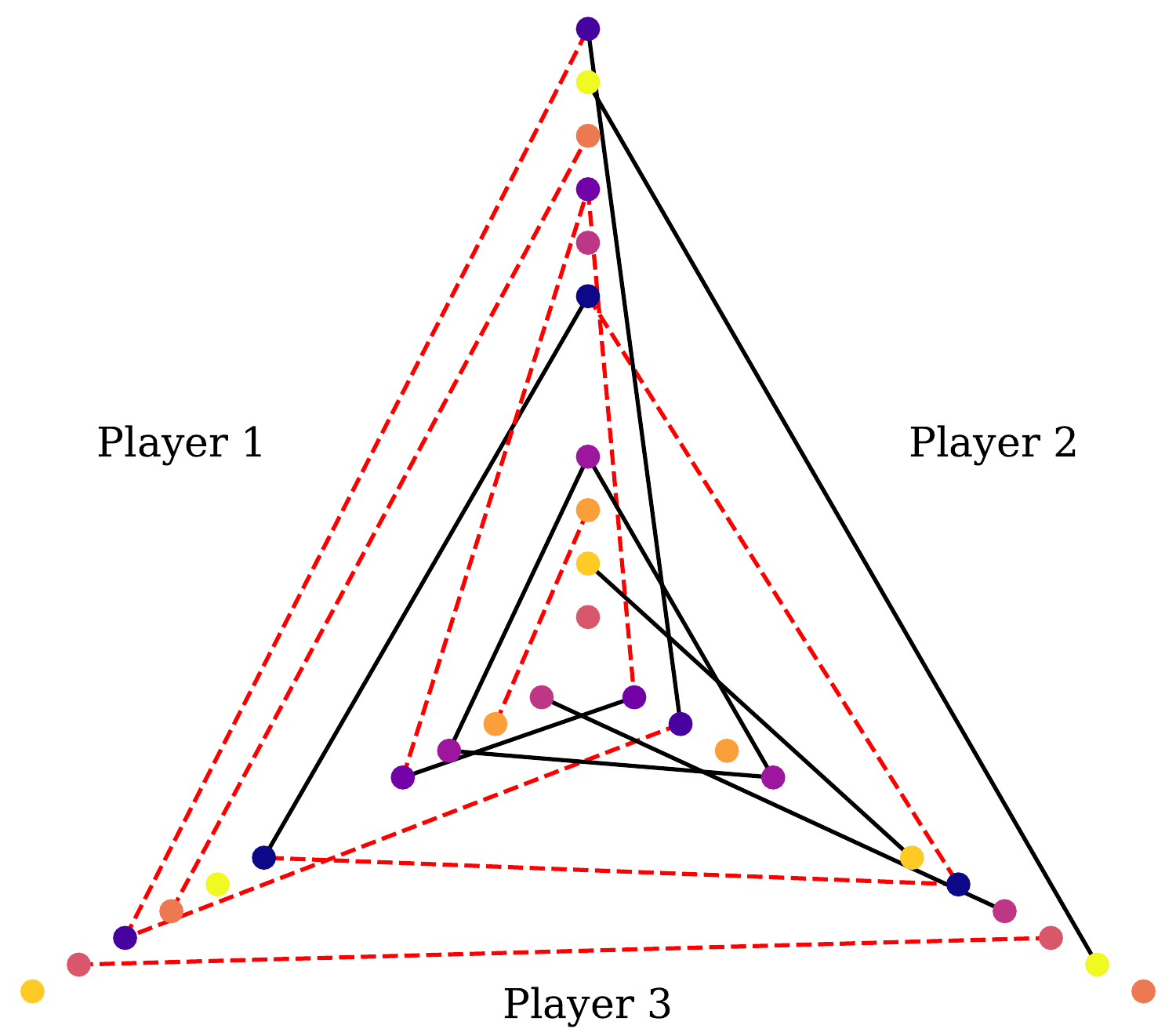}
      \caption{In the hard distribution, we randomly permute the
        vertices on top, on left, and on right.  Each player sees
        their edges, with associated labels, in a random order; they
        do not see the pre-permutation vertex identities (represented by color).}
      \label{fig:fig1b}
    \end{subfigure}
    \caption{Lower bound instance for triangle counting}\label{fig:instance}
\end{figure}

We will show that a uniformly random instance of this problem requires
$\Omega(n/T^{2/3})$ communication for the referee to succeed $2/3$ of
the time.  At the same time, it directly reduces to triangle counting:
each player sketches their edges labeled $0$ and sends the sketch to
the referee.  The referee adds the linear sketches up to get a sketch of all
$0$-labeled edges in the graph.  This subgraph either contains zero
triangles (if every triangle has an odd number of 1s) or very close
to $T/4$ (otherwise), so successfully counting triangles will
distinguish the two cases.

Our lower bound for the communication problem consists of two main
pieces.  First, we give a combinatorial lemma that bounds the players'
ability to co-ordinate any assignment of ``weights'' to edges or subsets of
edges, based on the structure of the graph.  Then we use Fourier-analytic
techniques to extend this to a lower bound of the communication required
by any protocol for the problem.

\paragraph{Combinatorial lemma.}  One approach the players could take
to solve the problem would be for each player to look at their $n$
edges and pick a $p$ fraction to send to the referee.  If the referee
receives a complete triangle, he can solve the problem.  What is the
expected number of triangles the referee receives, if the players
coordinate optimally?

The naive solution where players pick independently at random would
yield $p^3T$ triangles.  Vertex sampling---picking a $\sqrt{p}$
fraction of vertices, for example those of smallest index, and only
sampling edges between picked vertices---increases this to $p^{3/2}T$.
In~\cite{KP17}, a simple counting argument showed that ``oblivious''
strategies, which decide whether to sample an edge based only on the
edge and not the rest of the player's input, cannot do better than this.

The combinatorial lemma we need for the Fourier-analytic proof is a
stronger, generalized version of this sampling lemma.  It considers
players that receive some private randomness $\psi_e$ and
partially-shared randomness $\phi_u$ for each $u \in e$, and output an
arbitrary deterministic function
\[
  g_e = g_e(\psi_e, (\phi_u)_{u \in e}) \in [-1, 1]
\]
of their inputs.  If the $\phi_u$ are fully independent and the
$\psi_e$ are $(\abs{E}-1)$-wise independent, and
\[
  \max_e \E[\psi,\phi]{g_e^2} \leq p
\]
for some $p$, then we show:
\begin{align}
  \E[\psi,\phi]{\prod_e g_e} \lesssim p^{3/2}.\label{eq:lemgoal}
\end{align}
To relate this to sampling, we note that the communication
problem in Figure~\ref{fig:instance} can be constructed with randomness 
in the form above: $\phi_u$ contains the permutation of the vertices
$V_u$ associated with $u$, and $\psi_e$ contains player $e$'s edge
labels $x_e$ (which are $2$-wise independent) and the random order
$\pi_e$ in which they see their edges.  Consider picking a random
triangle edge $s \in [T]$ and adding to $\psi_e$ the index of $s$ in
player $e$'s list. (One can show that $\psi_e$ remains pairwise
independent, despite the shared dependence on $s$.)  If we only allow
$g_e \in \{0, 1\}$, then we can think of $g_e$ as the event that
player $e$ samples their edge in the $s$th triangle.  The condition on
$\E{g_e^2}$ says that each player can pick at most a $p$ fraction of
their edges, on average over their inputs.  The conclusion is that the
expected fraction of triangles completely sampled is at most
$p^{3/2}$.

This combinatorial lemma is different from the simple sampling lemma
of~\cite{KP17} in several ways.  First, it allows players to look at
their entire inputs before deciding which edges to keep.  Second, while
the lemma of~\cite{KP17} was based on defining a fixed subset of edges to
keep (so the number kept depended only on which edges were seen), in our lemma the
players only need to keep a $p$ fraction of their inputs \emph{on
  average}, but the players do not have shared randomness.  If they
had shared randomness, there would be a trivial counterexample: with
probability $p$ every player samples every edge, giving the referee
$p T$ triangles in expectation.  Without shared randomness, they can
still use the correlation of their input for nontrivial algorithms:
for example, for $p = 2^{1-n}$ they can send their entire input if
every edge has the same label; because of the promise, if two players
sample their inputs then the third is much more likely to.  But this
coordination is less effective than vertex sampling.

The combinatorial lemma is also more general, in ways that are
important for the Fourier-analytic component of the proof.  It allows
for ``fractional'' choices of edges $g_e \in [-1, 1]$, with an
$\ell_2$ constraint.  This allows for alternative competitive
strategies---for example, placing $\sqrt{p}$ weight on every edge also
yields $p^{3/2}$---but no strictly better ones.  Additionally, the
lemma will extend to cases where instead of placing weight on individual
edges, the players place weight on \emph{sets} of $k$ triangle edges for some $k \geq 1$.
These will correspond to weight $k$ Fourier coefficients.

\paragraph{Fourier-analytic argument.} Our approach for lower
bounding the communication problem is inspired by \cite{GKKRd07}.  Let
$x_e \in \{0, 1\}^n$ be the player $e$'s labels before permutation
(i.e., from Figure~\ref{fig:fig1a}).  We consider the referee's
posterior distribution $p$ on the triangle parities,
$x_1^{1:T} \oplus x_2^{1:T} \oplus x_3^{1:T}$.  $p$ is supported on $\{0^T, 1^T\}$, and our
goal is to show that it is nearly uniform.

First, we express the referee's posterior distribution on the labels
$x = (x_e)_{e \in E}$.  Let $f_e(y) = 1$ if player $e$'s message to the referee is
consistent with $x_e = y$, and 0 otherwise, and let
$f: \{0, 1\}^{\abs{E}n} \to \{0, 1\}$ be given by $f(x) = \prod_e f_e(x_e)$.
The referee has two constraints on $x$: the message consistency
constraint $f(x)$, and a parity constraint $q(x)$.  His posterior
distribution is uniform on $\supp(fq)$.

The first observation we make relates the referee's total variation
distance to the Fourier spectrum of $fq$.  For indicator functions
$g: \{0, 1\}^m \to \{0, 1\}$ it makes sense to consider the renormalized
Fourier transform
\[
  \wt{g}(s) := \frac{2^{m}}{\abs{\supp(g)}} \wh{g}(s) = \E[x \in \supp(g)]{(-1)^{s \cdot x}}.
\]
With this normalization, we observe that
\[
  \Delta := \|p - \mathcal{U}(\{0^T, 1^T\})\|_{TV} = \frac{1}{2}\wt{fq}(e_1, e_1, e_1)
\]
where $e_1 = (1, 0, \dotsc, 0) \in \{0, 1\}^n$.  Using the structure
of $q$'s spectrum and the Fourier convolution theorem, we turn this into
\[
\Delta = C \sum_{\substack{t \in \{0, 1\}^T\\\abs{t} \equiv 1 \mod 2}} \prod_e \wt{f_e}(t 0^{n-T})
\]
where $C$ is a normalising factor that is constant in expectation over $x_1, x_2, x_3$.  The combinatorial
lemma applied to $\wt{f_e}$ lets us bound the sum for a fixed
$\abs{t} = k$ (in expectation over the input).  The bound is, for some constant $D > 0$,
\[
 \sum_{\substack{t \in \{0, 1\}^T\\\abs{t} = k}} \prod_e \wt{f_e}(t 0^{n-T}) \le D\binom{T}{k}\left(\max_e \frac{1}{\binom{n}{k}}\E{\sum_{\substack{s \in \{0, 1\}^n\\\abs{s} = k}} \wt{f_e}(s)^2}\right)^{3/2}
\]
which can be bounded in terms of the players' $c$ bits of
communication by the KKL lemma (for small $k$) and Parseval's
identity (for high $k$). See Section \ref{sec:fourierfacts} for statements of the bounds used.  The dominant term when summing over $k$ is
$k=1$, whence we get that the referee's total variation distance has
\[
  \E{\Delta} \lesssim T(c/n)^{3/2}.
\]
This implies the players must send at least $n/T^{2/3}$ bits to
distinguish the two cases with significant probability.

\paragraph{Changes for non-triangle graphs.}
For counting general (hyper)graphs, the combinatorial lemma as
described gives a bound of $p^{MVC_{1/2}(H)}$, where $MVC_{1/2}(H)$ is
a ``modified'' fractional vertex cover in which weight can be placed
directly on edges for half price.  For odd cycles such as triangles,
$MVC_{1/2}(H)$ equals the non-modified fractional vertex cover $\tau$,
giving the desired $\Omega(n/T^{1/\tau})$ bound.

For other graphs, such as the length-3 path, $MVC_{1/2}(H)$ can be
less than $\tau$ leading to a suboptimal result.  For these graphs we
use a somewhat different proof, in which the referee is identified with
a particular edge $e^*$ in the graph.  The other players' inputs are
then completely independent of one another, with no promise on the XOR of their
labels. We follow a slightly different Fourier-analytic approach that requires
bounding
\[
  \E{\prod_{e \neq e^*} \wt{f_e}^2}.
\]
rather than $\E{\prod_{e} \wt{f_e}}$.  We apply the combinatorial lemma to
the $\wt{f_e}^2$, on which we have an $\ell_1$ constraint,
giving us the bound
\[
  \E{\prod_{e \neq e^*} \wt{f_e}^2} \leq p^{MVC_1(H \setminus e^*)}
\]
where the exponent is the non-modified fractional vertex cover of
$H \setminus e^*$. This gives a lower bound of
$\Omega(n/T^{1/MVC_1(H \setminus e^*)})$.  For every connected (non-hyper-)graph
$H$ that is neither an odd cycle nor a single edge, this equals
$\Omega(n/T^{1/\tau})$ for at least one $e^*$.

For graphs that \emph{are} single edges or odd cycles, $MVC_{1/2}(H) = \tau$, and so
the combination of these bounds gives
$\Omega(n / T^{1/\tau})$ for every connected graph with more than one
edge.  For hypergraphs, the individual lower bounds still hold, but
their maximum is not necessarily $\Omega(n / T^{1/\tau})$.

\paragraph{Dependence on $\eps$.}  The above approach is a lower bound
for distinguishing $T$ triangles from $0$ triangles.  Distinguishing
$T$ triangles from $(1-\eps)T$ triangles should require more space for small $\eps$. 
In the non-promise version of the proof used for graphs that
are not odd cycles, we use the noise operator, an operator that takes 
a binary function and ``noises'' it by randomly flipping input bits, 
to get an $\eps^{-2/\tau}$ dependence.  In the promise version, the bound we get
is only $\eps^{-1/\tau}$.

\subsection{Upper Bound}
For purposes of this overview, we describe a sampling algorithm for
the ``labeled'' version of the problem used in join size estimation,
where edges and vertices in $G$ correspond to edges and vertices in
$H$, and we only want to count subgraphs with matching labels.  Since
$H$ has constant size, we can solve the non-labeled version by trying
many random labelings.

Consider a hypergraph $H$ with minimal fractional vertex cover $f$, so
$f(u) \in [0, 1]$ for each vertex $u \in V_H$ and
$\sum_u f(u) = \tau$.  Let $\chi: V_G \to V_H$ be the labels.  For a
parameter $p \in (0, 1)$ to be determined later, we sample each vertex
$v \in V_G$ with probability $p^{f(\chi(v))}$, and we keep a hyperedge
$e \in E_G$ if and only if we sample all $v \in e$.

The chance we keep any given copy of $H$ is
$\prod_{u \in H} p^{f(u)} = p^{\tau}$.  Therefore, if we set
$p = 100/T^{1/\tau}$, the expected number of copies of $H$ we see will
be $p^{\tau}T \geq 100$.  On the other hand, the chance we keep any
single edge $e$ is $\prod_{v \in e} p^{f(\chi(v))} \leq p$, because
$f$ covers the edge associated with $e$ and so $\sum_{v \in e}f(\chi(v)) \ge 1$.  This gives a algorithm with
$O(mp) = O(m/T^{1/\tau})$ space that sees $p^\tau T \geq 100$ copies of
$H$ in expectation; from this $T$ can be estimated.

The only tricky bit is to show that the variance of the number of
sampled copies of $H$ is small.  We bound this in terms of the maximum
degree of $G$ and the maximum correlation between sampling two copies
of $H$ in $G$.  If this correlation is $1$ as can happen in general,
the sampling algorithm only works for constant-degree graphs.
However, if the vertex cover places at least $0.5$ weight on each
vertex of $H$, then the correlation is at most
$\sqrt{p} = \Theta(1/T^{1/(2\tau)})$.  This lets the algorithm work
for degree $O(T^{1/(2\tau)})$ graphs.  In the case of triangles, this
$O(T^{1/3})$ degree bound is the correct regime for $O(m/T^{2/3})$
samples to be possible---above this threshold, the maximum number of triangles
sharing a single vertex can be larger than $T^{2/3}$ and so the
$\Omega(m\sqrt{\Delta_V}/T)$ lower bound of~\cite{KP17} precludes it.

%!TEX root = ./main.tex

\section{Preliminaries}\label{sec:prelims}
\subsection{Roadmap}
This section will cover notation and certain basic facts about Boolean Fourier analysis. Section \ref{sec:weights} introduces a combinatorial lemma that will be needed for our lower bounds. Sections \ref{sec:promisegame} and \ref{sec:nopromisegame} will show lower bounds for two similar communication games, one where the players are given a promise on their inputs and one where they are not. Section \ref{sec:reduction} contains reductions from both of these problems to hypergraph counting, giving two non-comparable lower bounds, and a proof that these bounds combine for a tight bound in the case of non-hypergraphs. Finally, Section \ref{sec:upperbound} gives a counting algorithm with a matching upper bound.

\subsection{Notation}
We will write $e_i$ for the $n$-bit string $w$ such that $w_j = 1$ when $j = i$ and $0$ otherwise. When $x$ is an $n$ bit string and $A$ is a set, $(x)_{a \in A}$ will denote the $|A|$-tuple of strings given by repeating $x$ $|A|$ times. When there is a natural correspondence between elements of $A$ and subsets of $\lbrack |A|n \rbrack$, we will use tuples $(x_a)_{a \in A}$ interchangeably with $|A|n$ bit strings. If $x$ is a string, $x^{a:b}$ is the $(1 + b - a)$-bit substring consisting of the $a^\text{th}$ to the $b^\text{th}$ bit of $x$.

Let $H$ be a multi-hypergraph where empty edges are allowed, i.e.\ $H = (V, E)$, where $E$ is a multi-set of subsets of $V$. We define a modification to the standard fractional vertex cover wherein mass can be placed directly on edges, for some price:
\begin{definition}\label{def:MVC}
  For a weighted hypergraph $H = (V, E)$ with weights $w: E \to [0, \infty]$,
  we define the \emph{$\lambda$-modified fractional vertex cover number}
  $MVC_\lambda(H, w)$ to be:
  \[
    MVC_\lambda(H, w) = \min_f \left(\sum_{v \in V} f(v) + \lambda\sum_{e \in E} f(e)\right)
  \]
  over all $f: V \cup E \to [0, \infty]$ satisfying
  \begin{align*}
    \sum_{v \in e} f(v) + f(e) \ge w(e) \qquad \forall e \in E.
  \end{align*}
  When $w$ is omitted, it is assumed that $w(e) = 1$ for all $e$.  We
  note that $MVC_\lambda(H)$ equals the standard fractional vertex
  cover number of $H$ whenever $\lambda \geq 1$ and $H$ has no empty
  hyperedges.
\end{definition}

\subsection{Basic Facts About Boolean Fourier Analysis}
\label{sec:fourierfacts}
\begin{definition}
Let $f : \lbrace 0, 1\rbrace \rightarrow \mathbb{R}$. The \emph{Fourier transform} $\wh{f} : \lbrace 0, 1\rbrace \rightarrow \mathbb{R}$ of $f$ is given by: \[
\wh{f}(s) = \frac{1}{2^n}\sum_{z \in \lbrace 0, 1\rbrace^n} f(z) \chi_s(z)
\]
Where $\chi_s(z) = (-1)^{s\cdot z}$.
\end{definition}

\begin{lemma}
\label{lem:convolution}
Let $f,g : \lbrace 0, 1\rbrace^n \rightarrow \mathbb{R}$ be functions. Then: \[
\wh{fg}(s) = \sum_{t \in \lbrace 0, 1\rbrace^n}\wh{f}(s)\wh{g}(s \oplus t)
\]
\end{lemma}
\begin{proof}
See section 2.3 of \cite{W08}.
\end{proof}

\begin{lemma}
\label{lem:fdecomp}
Let $f: \lbrace 0, 1\rbrace^{kn} \rightarrow \mathbb{R}$ be given by: $f((z_i)_{i = 1}^k) = \prod_{i=1}^kf_i(z_i)$
for functions $f_i: \lbrace 0, 1\rbrace^{n} \rightarrow \mathbb{R}$.  Then, for any $(s_i)_{i = 1}^k \in \lbrace 0, 1\rbrace^{kn}$: \[
\widehat{f}((s_i)_{i = 1}^{k}) = \prod_{i = 1}^k\widehat{f_i}(s_i)
\]
\end{lemma}
\begin{proof}
See Appendix~\ref{app:a}.
\end{proof}

\noindent
One of our two main tools for bounding sums of Fourier coefficients will be Parseval's identity:
\begin{lemma}[Parseval]
\label{lem:parseval}

For every function $f : \lbrace 0,1\rbrace^n \rightarrow \mathbb{R}$ we have: \[
\sum_{z \in \lbrace 0, 1 \rbrace}f(z)^2 = 2^n\sum_{s \in \lbrace 0, 1\rbrace^n}\widehat{f}(s)^2
\]
\end{lemma}

\noindent
The other will be the KKL lemma:

\begin{lemma}[\cite{KKL88}]
\label{lem:kkl}
Let $f$ be a function $f : \lbrace 0, 1 \rbrace^n \rightarrow \lbrace -1, 0, 1 \rbrace$. Let $A = \lbrace x | f(x) \not = 0 \rbrace$, and let $s$ denote the Hamming weight of $s \in \lbrace 0,1\rbrace^n$. Then for every $\delta \in \lbrack 0, 1\rbrack$ we have \[\
\sum_{s \in \lbrace 0, 1\rbrace^n} \delta^{|s|}\widehat{f}(s)^2 \le \left(\frac{|A|}{2^n}\right)^\frac{2}{1 + \delta}
\]
\end{lemma}
\noindent
We will make use of the following corollary of this lemma, similar to a corollary from \cite{GKKRd07}:
\begin{lemma}
\label{lem:kklcor}
For any set $A \subseteq \lbrace 0, 1\rbrace^{n}$ and $\lambda \in (1, \infty)$, let $f : \lbrace 0, 1\rbrace^{n} \rightarrow \lbrace 0, 1\rbrace$ be the characteristic function of $A$, and suppose that $|A| \ge 2^{n - c}$ for some $c \in \mathbb{N}$. Then, for each $k \in \lbrack \lfloor \lambda c\rfloor\rbrack$ one has $\frac{2^{2n}}{|A|^2}\sum_{s \in \lbrace 0,1\rbrace^{n}: |s| = k} \widehat{f}(s)^2 \le \left(\frac{2 \lambda c  }{k}\right)^k$.
\end{lemma}
\noindent
The proof closely follows Lemma 6\ in \cite{GKKRd07} and is given in Appendix~\ref{app:a} for completeness.

\begin{lemma}
\label{lem:setsizes}
Let $m : \lbrace 0, 1\rbrace^n \rightarrow \lbrace 0, 1\rbrace^l$ be a function, and let $l \le c - \alpha$ for some $\alpha > 0$, and let $X$ be uniformly distributed over $\lbrace 0, 1\rbrace^n$. Define $F = \lbrace z \in \lbrace 0, 1\rbrace^n : m(z) = m(X) \rbrace$. Then, with probability at least $1 - 2^{-\alpha}$ over $X$, \[
|F| \ge 2^{n - c}.
\]
\end{lemma}
\begin{proof}
As $m$ takes only $2^l$ different values, there are at most $\frac{1}{2^\alpha}2^{c}$ distinct possible values for $F$, which partition $\lbrace 0, 1\rbrace^{n}$, and so no more than a $\frac{1}{2^\alpha}$ fraction of strings in $\lbrace 0, 1\rbrace^{n}$ are in sets of size $\le 2^{n-c}$.  Therefore the probability of a random string being in such a set is $\le \frac{1}{2^\alpha}$.
\end{proof}

\noindent
Following \cite{W08}, we define the noise operator $\mathcal{T}_\varepsilon$ on functions of $n$-bit strings as follows: \[
\mathcal{T}_\varepsilon(f)(x) = \E[y]{f(y)}
\]
Where $y$ is the random variable obtained by, for each bit of $x$, independently flipping it with probability $1/2 - \varepsilon/2$. We will use the fact (also in \cite{W08}) that, for all $s \in \lbrace 0, 1\rbrace^n$: \[
\widehat{\mathcal{T}_\varepsilon(f)}(s) = \varepsilon^{|s|}\widehat{f}(s).
\]

%!TEX root = ./main.tex
 \section{The Weights Lemma}\label{sec:weights}
 \begin{definition}[Totally disconnected hypergraph]\label{def:disconnected} For a hypergraph $H=(V_H, E_H)$ we say that $H$ is totally
    disconnected if edges of $H$ are pairwise disjoint, i.e. for every $a, b \in E$ one has
    $a \cap b = \emptyset$.
 \end{definition}
  \begin{lemma}
    \label{lem:weights}
    Consider any hypergraph $H = (V, E)$ and weight function
    $w: E \to [0, \infty]$.  Suppose that $H$ is not totally
    disconnected as per Definition~\ref{def:disconnected}, i.e., there exist $a, b \in E$ such that
    $a \cap b \neq \emptyset$.

    Consider any collection of random variables $g_e$ (for $e \in E$)
    that can be expressed as deterministic functions of some random
    variables $\phi_u$ (for $u \in V$) that are independent, and
    $\psi_e$ (for $e \in E$) that are independent of the $\phi_u$ and
    $(\abs{E}-1)$-wise independent themselves, i.e.,
    \[
      g_e = g_e(\psi_e, (\phi_u)_{u \in e}).
    \]

    Let $q \in \{1, 2\}$ and $0 < p < 1$, and suppose for all $e \in E$ that
    $\abs{g_e} \leq 1$ always and that
    \[
      \E{\abs{g_e}^q} \leq p^{w(e)}.
    \]
    Then
    \[
      \E{\prod_{e \in E} g_e} \leq d^{|V|}p^{MVC_{1/q}(H, w)}
    \]
    where $MVC_{1/q}(H, w)$ is the modified fractional vertex cover
    number per Definition~\ref{def:MVC}, and $d$ is the maximum of $1$
    and the greatest degree of a vertex $v \in V$.
  \end{lemma}
  \begin{proof}
    We will induct on the number of vertices that lie in at least two
    edges of $E$.  By assumption, this is at least 1; let $u$ be one
    such vertex.

    Let $H^u = (V^u, E^u)$ denote the hypergraph obtained from $H$ by
    removing the vertex $u$ from $V$ and every edge in $E$.  Let
    $\kappa: E \to E^u$ be the mapping associated with this
    transformation, and define $\psi_{\kappa(e)} = \psi_e$.  For any
    given $\phi$, we define $g_{\kappa(e)}^\phi$ to be $g_e$
    conditioned on $\phi_u = \phi$:
  \[
    g_{\kappa(e)}^{\phi} = (g_e \mid \phi_u = \phi) = g_e(\psi_e, (\phi_v)_{v \in e}|_{\phi_u = \phi}).
  \]
  Let $\gamma$ denote $\E{\prod_{e \in E} g_e}$, the expectation we want to bound, and let
  $\Phi_{V}, \Psi_{E}$ denote the collection of $\phi_v$ and $\psi_e$,
  respectively.  We can then rewrite our desired quantity as
  \begin{align}
  \gamma = \E[\phi_u]{\E[\Phi_{V^u}, \Psi_{E^u}]{\prod_{e \in E^u}g_{e}^{\phi_u}}}.\label{eq:gammastep2}
  \end{align}
  For any fixed $\phi$ we can define the modified weight function
  $w_\phi: E^u \to [0, \infty]$ that results from conditioning on
  $\phi_u = \phi$:
  \[
    w_\phi(e) := \log_p \E[\Phi_{V^u},\psi_e]{\abs{g_e^\phi}^q}.
  \]
  We know for all $e \in E$ that
  \begin{align}
    \E[\phi_u]{p^{w_{\phi_u}(\kappa(e))}} = \E[\Phi_V,\psi_e]{\abs{g_e}^q} \leq p^{w(e)}.\label{eq:pwavg}
  \end{align}
  Additionally, when $u \notin e$, $g_e$ is independent of $\phi_u$ so
  $w_\phi(\kappa(e))$ is independent of $\phi$, and hence
  $w_\phi(\kappa(e)) \geq w(e)$.

  We now proceed to show for every $\phi$ that the inner term
  in~\eqref{eq:gammastep2} satisfies
  \begin{align}
    \label{eq:reducedgoal}
    \E[\Phi_{V^u}]{\E[\Psi_E]{\prod_{e \in
          E^u}{g_{e}^{\phi}}}} \le d^{|V| -
        1}p^{MVC_{1/q}(H, w) - \max_{e \ni u} (w(e) - w_\phi(\kappa(e)))}.
  \end{align}
  For each $\phi$, we
  consider two cases:
  \begin{description}
  \item[Inductive step: $H^{u}$ is not totally disconnected.] In this case $u$
    was not the only vertex in $H$ that appears in at least two edges,
    then $H^u$ satisfies the constraints of our lemma and has fewer
    vertices that appear in at least two edges.  Furthermore, the
    random variables $g_{\kappa(e)}^\phi$ satisfy the constraints for
    the lemma with weight function $w_\phi$.  Therefore by the
    inductive hypothesis:
    \[
      \E[]{\prod_{e \in E^u}g_{\kappa(e)}^{\phi}} \le d^{|V| -
        1}p^{MVC_{1/q}(H^u, w_\phi)},
    \]
    and it suffices to estimate $MVC_{1/q}(H^{u}, w_\phi)$. As previously
    noted, every edge $e$ such that $w_{\phi}(\kappa(e)) < w(e)$
    contains $u$, so we can cover $(H, w)$ by taking any cover of
    $(H^u, w_\phi)$ and placing
    $\max_{e \ni u} (w(e) - w_{\phi}(\kappa(e)))$ weight on $u$.
    Hence
    \[
      MVC_{1/q}(H, w) \leq MVC_{1/q}(H^u, w_\phi) + \max_{e \ni u} (w(e) - w_{\phi}(\kappa(e)))
    \]
    which, with the previous equation, gives~\eqref{eq:reducedgoal}.
  \item[Base case: $H^{u}$ is totally disconnected.]  In this case,
    $u$ is the only vertex that appears in at least two edges of $E$.
    Let $E_1 = \{e \in E \mid u \in e\}$ and $E_2 = E \setminus E_1$.
    Let $e' = \argmax_{e \in E_1} w_\phi(e)$ and $e'' = \argmax_{e \in E_1 \setminus \lbrace e' \rbrace} w_\phi(e)$.  We note that
    \[
      MVC_{1/q}(H, w_\phi) = w_\phi(e'') + \frac{1}{q} (w_\phi(e') - w_\phi(e'')) + \frac{1}{q} \sum_{e \in E_2} w_\phi(e)
    \]
    because $q \in \{1, 2\}$.

    Let $h_e = g_{\kappa(e)}^\phi = (g_e \mid \phi_u = \phi)$, a
    function of $\psi_e$ and $(\phi_v)_{v \in e}$.  The $h_e$ for
    all $e \in E_2 \cup \{e'\}$ are independent of each other, because
    these are $\abs{E_2} + 1 \leq \abs{E} - 1$ variables, so the
    $\psi_{e}$ are fully independent, and no $\phi_v$ variable appears
    in more than one such $h_e$.  The LHS of~\eqref{eq:reducedgoal}
    which we want to bound is equal to
    \begin{align*}
      \E[\Phi_{V^u},\Psi_E]{\prod_{e \in E} h_e}
      &\leq \E[(h_e)_{e \in E}]{\prod_{e \in E_2 \cup \{e', e''\}} \abs{h_e}}\\
      &= \E[(h_e)_{e \in E_2}]{\prod_{e \in E_2} \abs{h_e} \E[h_{e'}, h_{e''}]{\abs{h_{e'}h_{e''}} \mid (h_e)_{e \in E_2}}}
    \end{align*}
    On the other hand, the dependency structure implies
    \[
      \E[]{\abs{h_{e'}}^q \mid (h_e)_{e \in E_2}} = \E{\abs{h_{e'}}^q} = p^{w_\phi(\kappa(e'))}
    \]
    regardless of the values of $h_e$ being conditioned upon.  By the
    same logic,
    \[
      \E[]{\abs{h_{e''}}^q \mid (h_e)_{e \in E_2}} = p^{w_\phi(\kappa(e''))}.
    \]
    Splitting into cases for $q \in \{1, 2\}$, we have by H\"older's
    inequality that
    \[
      \E{\abs{h_{e'} h_{e''}} \mid (h_e)_{e \in E_2}} \leq \left\{
        \begin{array}{rll}
          \E{h_{e'}^2}^{1/2}\E{h_{e''}^2}^{1/2} &= p^{\frac{1}{2} w_\phi(\kappa(e')) + \frac{1}{2} w_\phi(\kappa(e''))}& \text{for } q = 2\\
          \E{\abs{h_{e'}}} \cdot 1 &= p^{w_\phi(\kappa(e'))}& \text{for } q = 1
        \end{array}
      \right.
    \]
    In either case,
    \[
      \E{\abs{h_{e'} h_{e''}} \mid (h_e)_{e \in E_2}} \leq p^{(1 - \frac{1}{q})w_\phi(\kappa(e'')) + \frac{1}{q}w_\phi(\kappa(e'))}
    \]
    so the quantity we want to bound is
    \begin{align*}
      \E[]{\prod_{e \in E} h_e}
      &\leq p^{(1 - \frac{1}{q})w_\phi(\kappa(e'')) + \frac{1}{q}w_\phi(\kappa(e'))} \prod_{e \in E_2} \E{\abs{h_e}}\\
      &\leq p^{(1 - \frac{1}{q})w_\phi(\kappa(e'')) + \frac{1}{q}w_\phi(\kappa(e'))} \prod_{e \in E_2} \E{\abs{h_e}^q}^{1/q}\\
      &\leq p^{(1 - \frac{1}{q})w_\phi(\kappa(e'')) + \frac{1}{q}w_\phi(\kappa(e')) + \frac{1}{q}\sum_{e \in E_2} w(e)}\\
      &= p^{MVC_{1/q}(H, w)} \cdot p^{(1 - \frac{1}{q})(w_\phi(\kappa(e'')) - w(e'')) + \frac{1}{q}(w_\phi(\kappa(e')) - w(e'))}\\
      &\leq p^{MVC_{1/q}(H, w)} \cdot p^{\min_{e \in \{e', e''\}} (w_\phi(\kappa(e)) - w(e))}\\
      &\leq p^{MVC_{1/q}(H, w)} \cdot p^{-\max_{e \ni u} (w(e) - w_\phi(\kappa(e)))}
    \end{align*}
    giving~\eqref{eq:reducedgoal}.
  \end{description}
  Combining the two cases gives~\eqref{eq:reducedgoal} for all $\phi$ unconditionally.
  Plugging into~\eqref{eq:gammastep2} gives
  \begin{align*}
  \gamma &\le  d^{|V| - 1}p^{MVC_{1/q}(H, w)}\E[\phi_u]{p^{- \max_{e \ni u} (w(e) - w_{\phi_u}(\kappa(e)))}}\\
  &=  d^{|V| - 1}p^{MVC_{1/q}(H, w)}\E[\phi_u]{\max_{e \ni u}p^{ w_{\phi_u}(\kappa(e)) - w(e)}}
  \end{align*}
  For any given $e$, from~\eqref{eq:pwavg} we have
  \begin{align*}
  \E[\phi_u]{p^{ w_{\phi_u}(\kappa(e)) - w(e)}} \leq p^{w(e)} \times p^{-w(e)} = 1.
  \end{align*}
  Since no more than $d$ edges include $u$, this means
  \[
  \E[\phi_u]{\max_{e \ni u}p^{ w_{\phi_u}(\kappa(e)) - w(e)}} \le \E[\phi_u]{\sum_{e \ni u}p^{ w_{\phi_u}(\kappa(e)) - w(e)}} \leq d 
  \]
  which gives
  \[
    \gamma \le d^{|V|}p^{MVC_{1/q}(H, w)}
  \]
  as desired.
\end{proof}

%!TEX root = ./main.tex
\section{Hypergraph Counting with a Promise}\label{sec:promisegame}
In both of the games that follow, we will assume that the players are deterministic. This is without loss of generality by Yao's minimax principle since the inputs are sampled from a fixed distribution.
\subsection{Game}
We will define a $|E| + 1$-player game $\mathtt{PromiseCounting}(H,n,T,\varepsilon)$ (with $H$ a hypergraph, $n, T \in \mathbb{N}, \varepsilon \in \lbrace 1/T, 2/T, \dots, 1\rbrace)$ , as follows: There is one referee, who receives messages from every other player. No other communication takes place. Each player besides the referee corresponds to an edge $e \in E$. 

Let $N = T + (n - T)|E|$. For each edge $e \in E$, let $L_e \subset \lbrack N \rbrack$ be an $n$-element set containing $\lbrack T \rbrack$ and $n - T$ elements disjoint from every other $L_e$, so that if $a \not = b$, $L_a \cap L_b = \lbrack T \rbrack$, and $\bigcup_{e \in E} L_e = \lbrack N \rbrack$. For each $e \in E$, let $\rho_e : \lbrack n \rbrack \rightarrow L_e$ be a fixed bijection such that $\rho_e \vert_{\lbrack T \rbrack}$ is the identity.

An instance of $\mathtt{PromiseCounting}(H,n,T,\varepsilon)$ is as follows:

\begin{itemize}
\item For each edge $e \in E$:
\begin{itemize}
\item A string $x_e \in \lbrace 0, 1\rbrace^n$.
\item A permutation $\pi_e$ on $L_e$.
\end{itemize}
\item For each vertex $v \in V$:
\begin{itemize}
\item A permutation $\pi_v$ on $\lbrack N \rbrack$.
\end{itemize}
\item A string $\tau \in \lbrace 0^{\varepsilon T}, 1^{\varepsilon T}\rbrace$
\end{itemize}

\noindent
The players have the following promise: \[
\bigoplus_{e \in E}x_e^{1:\varepsilon T} = \tau
\]
We will write $X$ for the strings $(x_e)_{e \in E}$, $\Pi_E$ for the permutations $(\pi_e)_{e \in E}$, and $\Pi_V$ for the permutations $(\pi_v)_{v \in V}$.
They have access to the following information:
\begin{itemize}
\item For each player $e \in E$:
\begin{itemize}
\item $x_e \rho_e \pi_e$
\item $(\pi_v(\pi_e^{-1}(i)))_{v \in e, i \in L_e}$
\end{itemize}
\item For the referee:
\begin{itemize}
\item $\Pi = (\Pi_E, \Pi_V)$
\end{itemize}
\end{itemize}

\noindent
Given the messages received from the players, the referee's task will be to determine whether $\tau = 0^{\varepsilon T}$ or $\tau = 1^{\varepsilon T}$. 

\subsection{Hard Instance}
We will lower bound the complexity of this problem under the following hard input distribution: $\tau$ is chosen uniformly from $\lbrace 0^{ \varepsilon T}, 1^{\varepsilon T}\rbrace$, and then the strings $(x_e)_{e \in E}$ are chosen uniformly from: \[
\left\lbrace (x_e)_{e \in E} \in \lbrace 0, 1 \rbrace^{|E|}: \bigoplus_{e \in E} x_e^{1: \varepsilon T} = \tau \right\rbrace
\]
Every permutation $\pi_u, \pi_e$ is chosen uniformly at random and independently of each other and the strings.

\begin{figure}
  \centering
    \begin{subfigure}[t]{0.47\textwidth}
      \centering
      \includegraphics[width=\textwidth]{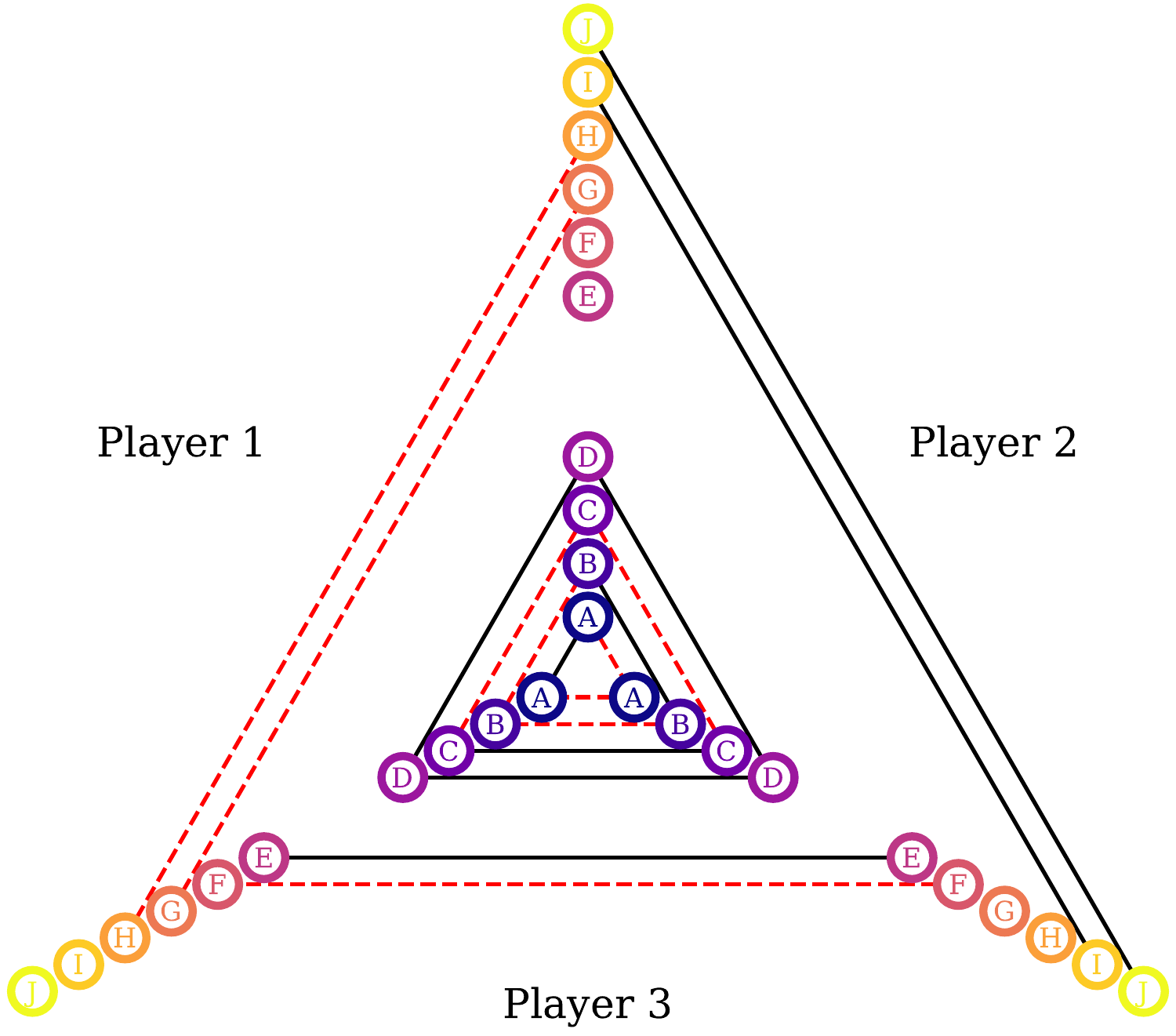}
      \caption{The player's instance ignoring the permutations.  The
        $x_e$ are the indices of red edges, read from inside out:
        $x_1 = [0,1,1,0,1,1]$, $x_2 = [1,0,1,0,0,0]$,
        $x_3 = [1,1,0,0,0,1]$}
      \label{fig:fig2a}
    \end{subfigure}\hfill
    \begin{subfigure}[t]{0.47\textwidth}
      \centering
      \includegraphics[width=\textwidth]{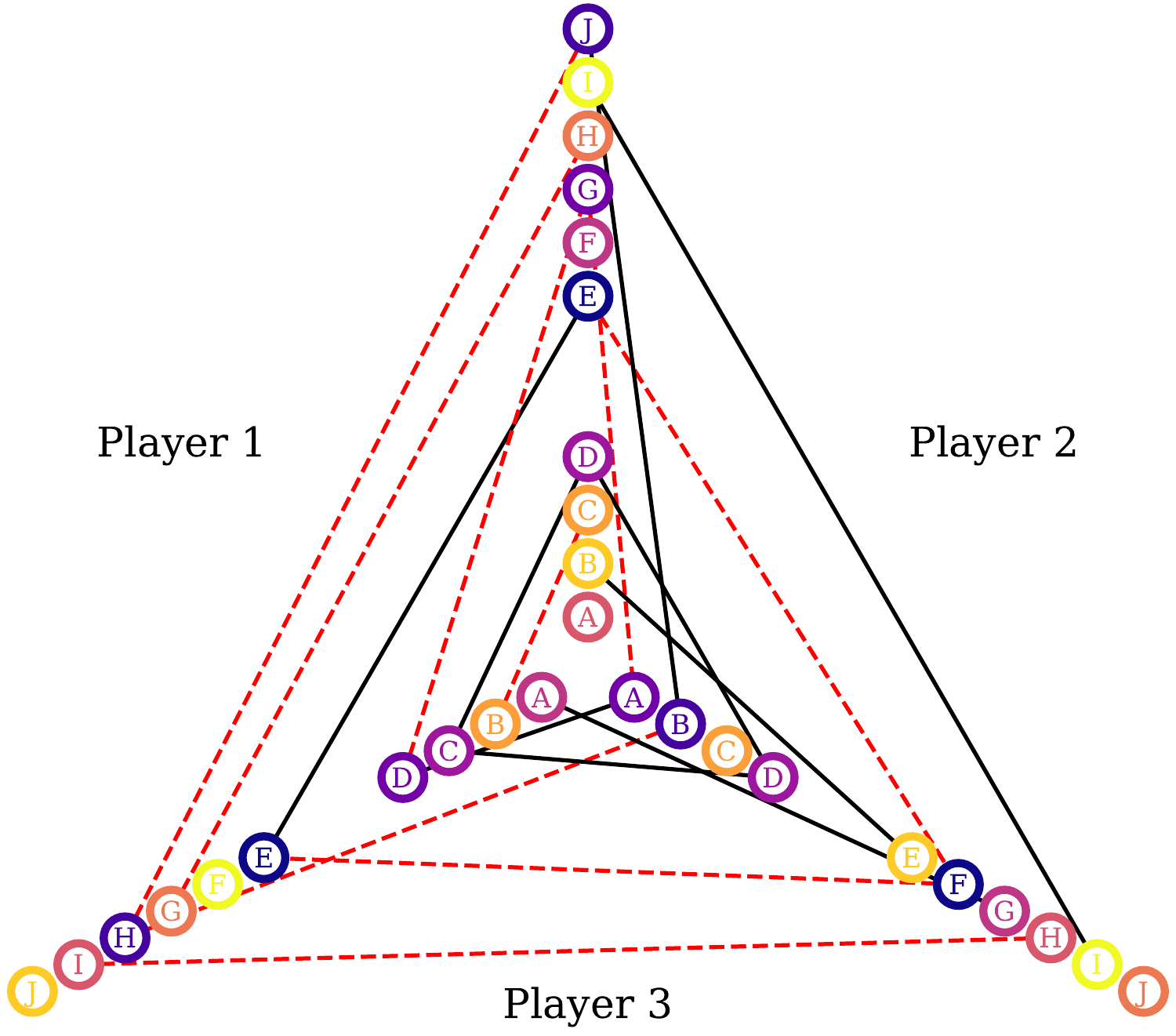}
      \caption{The hard distribution permutes each set of vertices.
        The players see their edges and associated labels, but not the
        vertex colors (which represent the pre-permutation
        identities).}
      \label{fig:fig2b}
    \end{subfigure}

    \vspace{1em}
    \begin{subfigure}[t]{0.47\textwidth}
      \centering
\begin{tabular}{|l|l|l||l|l|l||l|l|l|}
\hline
\multicolumn{3}{|c||}{Player 1} & \multicolumn{3}{|c||}{Player 2} & \multicolumn{3}{c|}{Player 3} \\
\hline
x&u&v&x&u&v&x&u&v\\
\hline\hline
0&E&E&0&B&E&0&D&C\\
\hline
1&H&J&0&D&D&0&G&A\\
\hline
1&G&H&1&E&F&1&H&I\\
\hline
1&B&C&0&I&I&0&A&D\\
\hline
1&D&G&1&G&A&1&F&E\\
\hline
0&C&D&0&J&B&1&B&H\\
\hline
\end{tabular}
      \caption{Each player's input consists of their edges
        in~(\subref{fig:fig2b}) in a random order.  $u$ represents the
        vertex counterclockwise of the player, and $v$ represents the vertex clockwise.}
      \label{fig:fig2c}
    \end{subfigure}
    \caption{Encoding of lower bound instance for triangle counting}\label{fig:instance2}
\end{figure}

\subsection{Lower Bound}
For each player $e$, write $m_e(x_e\rho_e\pi_e, (\pi_v(\pi_e^{-1}(i)))_{v \in e, i \in L_e})$ for the message the player sends to the referee on seeing $x_e\rho_e\pi_e$ and $(\pi_v(\pi_e^{-1}(i)))_{v \in e, i \in L_e}$.
 
\begin{theorem}\label{thm:lb-promise}
  Let $H$ be a connected hypergraph with more than one edge. Let
  $c \in \lbrack n \rbrack$. Suppose that, for all inputs $(X, \Pi)$
  to the game, no player sends a message of more than $c$ bits, and suppose that $\varepsilon T \le n/10$.

Let $p : \lbrace 0^{\varepsilon T}, 1^{\varepsilon T} \rbrace \rightarrow \lbrack 0, 1\rbrack$ be the referee's posterior distribution on $\tau$. Let $\nu$ be the distribution of $\mathcal{U}(\lbrace 0^{\varepsilon T}, 1^{\varepsilon T} \rbrace)$, the uniform distribution on the two-element set $\lbrace 0^{\varepsilon T}, 1^{\varepsilon T}\rbrace$. Let $\mu = MVC_{1/2}(H)$, and let $0 < \delta < 1$.

There exists a constant $\gamma$, depending only on $H$, such that, if $c \le \gamma \frac{n}{(\delta^2 \varepsilon T)^{1/\mu}}$:
\[
\E[X, \Pi]{||p - \nu||_{TV}} \le \delta
\]
\end{theorem}
\noindent
We will prove a weaker form of the theorem in which no player sends a
message of more than $c - 2 \log(3\abs{E}/\delta) - C'$ bits, for a
sufficiently large constant $C'$.  This implies the lemma statement
for a slightly larger $\gamma$, since the adjustment is
$O(\frac{n}{(\delta^2\varepsilon T)^{1/\mu}})$.

We will prove this by relating the distance of $p$ from uniform to the Fourier coefficients of the indicator functions associated with the messages sent by the players.

For each $e \in E$, define the random function $f_e : \lbrace 0, 1\rbrace^n \rightarrow \lbrace 0, 1 \rbrace$ by: \[
f_e(z) = \begin{cases}
1 & \mbox{if $m_e(z\rho_e\pi_e, (\pi_v(\pi_e^{-1}(i)))_{v \in e, i \in L_e}) = m_e(x_e\rho_e\pi_e, (\pi_v(\pi_e^{-1}(i)))_{v \in e, i \in L_e})$.}\\
0 & \mbox{otherwise.}
\end{cases}
\]
And, with $Y = ((y_e)_{e\in E})$, define $f : \lbrace 0, 1\rbrace^{|E|n} \rightarrow \lbrace 0, 1 \rbrace$ by: \[
f(Y) = \prod_{e \in E}f_e(y_e)
\]
Let $q : \lbrace 0, 1 \rbrace^{|E|n} \rightarrow \lbrace 0, 1 \rbrace$ be given by: \[
q(Y) = \begin{cases}
1 & \mbox{if $\bigoplus_{e \in E} y_e^{1:\varepsilon T} \in \lbrace 0^{\varepsilon T}, 1^{\varepsilon T}\rbrace$.}\\
0 & \mbox{otherwise.}
\end{cases}
\]
Let $F_e = f_e^{-1}(\lbrace 1\rbrace) \subseteq \{0,1\}^n$, and let: \[
F = \prod_{e \in E} F_e \subseteq \{0, 1\}^{\abs{E}n}.
\]  
Let \[
Q = q^{-1}(\lbrace 1\rbrace) \subseteq \{0, 1\}^{\abs{E}n}
\]
and \[
J = F \cap Q.
\]
So then, as the game's inputs are uniformly distributed among those such that $\bigoplus_{e \in E}x_e^{1:\varepsilon T} \in \lbrace 0^{\varepsilon T}, 1^{\varepsilon T}\rbrace$, the referee's posterior distribution on $X$ is uniform on $J$. We are now ready to calculate the referee's posterior distribution on $\tau$, writing $Z = (z_e)_{e \in E}$. 
\[
p(y) = \frac{\left|\left\lbrace Z \in J : \bigoplus_{e \in E} z_e^{1:\varepsilon T} = y\right\rbrace\right|}{|J|}
\]
We can now address the total variation distance:
\begin{align*}
||p - \nu||_{TV} & = \frac{1}{2}|p(0^{\varepsilon T}) - p(1^{\varepsilon T})|\\
&= \frac{1}{2|J|} \left|\left(\left|\left\lbrace Z \in J : \bigoplus_{e \in E} z_e^{1:{\varepsilon T}} = 0^{\varepsilon T}\right\rbrace\right| - \left|\left\lbrace Z \in J : \bigoplus_{e \in E} z_e^{1:{\varepsilon T}} = 1^{\varepsilon T}\right\rbrace\right|\right)\right|\\
&= \frac{1}{2|J|}\sum_{Z \in \lbrace 0, 1\rbrace^{|E|n}}f(z)q(z)(-1)^{\sum_{e \in E}(z_e)_1}\\
&= \frac{2^{|E|n-1}}{|J|}\widehat{fq}((e_1)_{e \in E})
\end{align*}

\noindent
We now introduce a couple of lemmas characterizing the Fourier coefficients of $fq$.

\begin{lemma}
\[
\widehat{q}(S) = \begin{cases}
2^{1-{\varepsilon T}} & \mbox{if $S = (s0^{n-{\varepsilon T}})_{e \in E}$ with $|s|$ even.}\\
0 & \mbox{otherwise.}
\end{cases}
\]
\end{lemma}
\begin{proof}
First, suppose that $S = (s_e)_{e \in E}$ is not of the form $(s)_{e \in E}$ for some $s \in \lbrace 0, 1\rbrace^n$. Then there exist $a, b \in E, i \in \lbrack n\rbrack$ such that $(s_a)_i = 1$ and $(s_b)_i = 0$. Partition the strings $z \in \lbrace 0, 1\rbrace^{|E|n}$ into pairs $z, \wt{z}$ by defining $\wt{z}$ to be $z$ with $(z_a)_i$ and $(z_b)_i$ flipped. Now, $\bigoplus_{e \in E}z_e^{1:{\varepsilon T}} = \bigoplus_{e \in E}\wt{z}_e^{1:{\varepsilon T}}$, so $q(z) = q(\wt{z})$, while $\chi_{S}(z) = -\chi_{S}(z)$, so $q(z)\chi_S(z) + q(\wt{z})\chi_S(\wt{z}) = 0$. Therefore:
\begin{align*}
\widehat{q}(S) &= \frac{1}{2^{|E|n}}\sum_{z \in \lbrace 0, 1\rbrace^{|E|n}}q(z)\chi_S(z)\\
&= 0
\end{align*}
Now, suppose that $S = (s)_{e \in E}$ for some $s \in \lbrace 0, 1\rbrace^n$ with $s^{{\varepsilon T}+1:n} \not = 0^{n-{\varepsilon T}}$. Then, let ${\varepsilon T} < i \le n$ be such that $s_i = 1$. Choose some edge $e'$ in $E$ arbitrarily. Partition the strings $z \in \lbrace 0, 1\rbrace^{|E|n}$ into pairs $z, \wt{z}$ by defining $z$ to be $z$ with the $i^\text{th}$ bit of $z_{e'}$ flipped. Now, $\bigoplus_{e \in E}z_e^{1:{\varepsilon T}} = \bigoplus_{e \in E}\wt{z}_e^{1:{\varepsilon T}}$, so $q(z) = q(\wt{z})$, while $\chi_{S}(z) = -\chi_{S}(z)$, so $q(z)\chi_S(z) + q(\wt{z})\chi_S(\wt{z}) = 0$. Therefore:
\begin{align*}
\widehat{q}(S) &= \frac{1}{2^{|E|n}}\sum_{z \in \lbrace 0, 1\rbrace^{|E|n}}q(z)\chi_S(z)\\
&= 0
\end{align*}
Now suppose $S = (s0^{n-{\varepsilon T}})_{e \in E}$ for some $s \in \lbrace 0, 1\rbrace^{\varepsilon T}$ such that $|s|$ is odd. Choose some edge $e'$ in $E$ arbitrarily. Partition the strings in $\lbrace 0, 1\rbrace^{|E|n}$ into pairs $z, \wt{z}$  by defining $\wt{z}$ to be $z$ with the first through ${\varepsilon T}^{\text{th}}$ bits in $z_{e'}$ flipped. Then $\bigoplus_{e \in E}z_e^{1:{\varepsilon T}} = \bigoplus_{e \in E}{\wt{z}_e}^{1:{\varepsilon T}} + 1^{\varepsilon T}$, and so $q(\wt{z}) = q(z)$. However, $\chi_{S}(z) = -\chi_S(z)$, so $q(z)\chi_S(z) + q(\wt{z})\chi_S(\wt{z}) = 0$. Therefore:
\begin{align*}
\widehat{q}(S) &= \frac{1}{2^{|E|n}}\sum_{z \in \lbrace 0, 1\rbrace^{|E|n}}q(z)\chi_S(z)\\
&= 0
\end{align*}
Finally, suppose $S = (s0^{n - \eps T})_{e \in E}$ for some $s \in \lbrace 0, 1\rbrace^{\varepsilon T}$ such that $|s|$ is even. Then, for any $z$ such that $q(z) = 1$:
\begin{align*}
S\cdot z &= \sum_{e \in E}s0^{n-{\varepsilon T}} \cdot z_e\\
&\equiv s0^{n-{\varepsilon T}} \cdot \bigoplus_{e \in E} z_e \mod 2\\
&= \begin{cases}
s \cdot 0^{\varepsilon T}\\
s \cdot 1^{\varepsilon T}
\end{cases}\\
&\equiv  0 \mod 2
\end{align*}
So $\chi_S(z) = 1$ for all $z$ such that $q(z) = 1$. Therefore: 
\begin{align*}
\widehat{q}(z) &= \frac{|q^{-1}(\lbrace 1 \rbrace)|}{2^{|E|n}}\\
&= \frac{1}{2^{|E|n}} \left|\left\lbrace (z_e)_{e \in E} \in \lbrace 0, 1\rbrace^{|E|n} \middle| \bigoplus_{e \in E} z_e^{1:{\varepsilon T}} = 0^{\varepsilon T} \vee  \bigoplus_{e \in E} z_e^{1:{\varepsilon T}} = 1^{\varepsilon T} \right\rbrace\right|\\
&= \frac{1}{2^{|E|n}} 2|\lbrace 0, 1\rbrace^{|E|n-{\varepsilon T}}|\\
&= 2^{1-{\varepsilon T}}
\end{align*}
as desired.
\end{proof}

\begin{lemma}
\label{lem:qmult}
For any $f : \lbrace 0, 1 \rbrace^{|E|n} \rightarrow \lbrace 0, 1\rbrace$ and $z \in \lbrace 0, 1\rbrace^n$: \[
\widehat{qf}((s)_{e \in E}) = 2^{1-{\varepsilon T}}\sum_{\substack{t \in \lbrace 0, 1\rbrace^{\varepsilon T}\\|t| \equiv 1 \bmod 2}}\widehat{f}((ts^{({\varepsilon T}+1):n})_{e \in E})
\]
\end{lemma}
\begin{proof}
\begin{align*}
\widehat{qf}((z)_{e \in E}) &= \sum_{y \in \lbrace 0, 1\rbrace^{|E|n}}\widehat{f}((z)_{e \in E} \oplus y)q(y)\\
&= 2^{1-\varepsilon T}\sum_{\substack{t \in \lbrace 0, 1\rbrace^{\varepsilon T}\\ |t| \equiv 0 \mod 2}} \widehat{f}((z)_{e \in E} \oplus (t0^{n - \varepsilon T})_{e \in E})\\
&= 2^{1-{\varepsilon T}}\sum_{\substack{t \in \lbrace 0, 1\rbrace^{\varepsilon T}\\|t| \equiv 1 \bmod 2}}\widehat{f}((ts^{({\varepsilon T}+1):n})_{e \in E}).
\end{align*}
\end{proof}

\noindent
Applying Lemma \ref{lem:qmult}, our total variation bound becomes: \[
||p - \nu||_{TV} = \frac{2^{|E|n - \varepsilon T}}{|J|}\left|\sum_{\substack{t \in \lbrace 0, 1\rbrace^{\varepsilon T}\\|t| \equiv 1 \bmod 2}}\widehat{f}(t0^{n - \varepsilon T})\right|
\]
Then, by applying Lemma \ref{lem:fdecomp}, we can write:

\begin{align*}
||p - \nu||_{TV} &= \frac{2^{|E|n - \varepsilon T}}{|J|}\left|\sum_{\substack{t \in \lbrace 0, 1\rbrace^{\varepsilon T}\\|t| \equiv 1 \bmod 2}}\prod_{e \in E} \widehat{f_e}(t0^{n - \varepsilon T})\right|\\
&= \frac{\prod_{e \in E}|F_e|}{2^{\varepsilon T}|J|} \left|\sum_{\substack{t \in \lbrace 0, 1\rbrace^{\varepsilon T}\\|t| \equiv 1 \bmod 2}}\prod_{e \in E} \frac{2^n}{|F_e|}\widehat{f_e}(t0^{n - \varepsilon T})\right|
\end{align*}
  We will seek to bound this sum in expectation. To do so, we will use Lemma \ref{lem:weights} to show that the players cannot ``co-ordinate'' the Fourier coefficients of the functions $f_e$ well enough to make the above sum large.  In order to apply this lemma, we will need, for each $e$, a probabilistic bound on \[
\widehat{f_e}(s)^2.
\]
We will do this for each value of $k = |s|$, by using Lemma
\ref{lem:kklcor} when $k$ is close to $0$ or $n$, and Parseval's
identity for other $k$.  To apply Lemma \ref{lem:kklcor}, we need to
bound the size of $F_e$ from below. By applying Lemma
\ref{lem:setsizes} with $\alpha = \log \frac{3|E|}{\delta}$ and then using
the union bound, we have that with probability $\ge 1 - \delta/3$:
\[ \forall e \in E, |F_e| \ge 2^{n - c}.
\]
We also need the following bound on the normalizing factor:
\begin{lemma}
For any referee input $\Phi$: \[
\E[X|\Pi = \Phi]{\frac{\prod_{e \in E}|F_e|}{2^{\varepsilon T}|J|}} \leq 1.
\]
\end{lemma} 
\begin{proof}
Condition on $\Pi = \Phi$. Write $(F_e(Y))_{e \in E}, F(Y), J(Y)$ for the values $(F_e)_{e \in E}, F, J$ take when $X = Y$. We use the fact that, for any $\Pi$, the sets $\lbrace J(Y)\rbrace_{Y \in Q}$ partition $Q$, and likewise for each $e$, the sets $F_e(Y)$ partition $\lbrace 0, 1\rbrace^n$. Then note that, for any $Y, Z$, if $J(Y) \not= J(Z)$, $\prod_{e \in E} F_e(X)$ is disjoint from $\prod_{e\in E}F_e(Y)$, as they must be either disjoint or identical, and if they are identical so are $J(Y)$ and $J(Z)$.

So then, writing $\mathcal{J}$ for the set of distinct possible values of $J(X)$, and $R(J)$ for an arbitrary representative element of $J$: 
\begin{align*}
\E[X]{ \frac{\prod_{e \in E}|F_e(X)|}{|J(X)|}} &= \frac{1}{|Q|}\sum_{J \in \mathcal{J}}|J|\frac{\left|\prod_{e \in E} F_e(R(J))\right|}{|J|}\\
&= 2^{\varepsilon T - |E|n - 1} \sum_{J \in \mathcal{J}}\left|\prod_{e \in E} F_e(R(J))\right|\\
&\le 2^{\varepsilon T - |E|n - 1}|\lbrace 0, 1\rbrace^{|E|n}|\\
&= 2^{\varepsilon T}
\end{align*}
as desired.
\end{proof}

\noindent
We define $\mathcal{E}$ to be the event that \[
\forall e, |F_e| \ge 2^{n - c} \text{ and } \frac{\prod_{e \in E}|F_e|}{2^{\varepsilon T}|J|} \le \frac{3}{\delta}
\]
which happens with probability at least $1 - 2\delta/3$.

We now define a renormalized and masked version of $\wh{f_e}$ as follows:
\[
\wt{f_e}(s) := \begin{cases}
\frac{2^{n}\widehat{f_e}(s)}{|F_e|} & \mbox{if $|F_e| \ge 2^{n-c}$.}\\
0 & \mbox{otherwise.}
\end{cases}
\]
Note that $\wt{f_e}$ can be expressed as a \emph{deterministic} function of
the following form:
\[
\wt{f_e}(s) = g_e(s\rho_e\pi_e, x_e\rho_e\pi_e, (\pi_v(\pi_e^{-1}(i)))_{v \in e, i \in L_e})
\]
To justify this, first recall that $F_e$ is determined by $(x_e\rho_e\pi_e, (\pi_v(\pi_e^{-1}(i)))_{v \in e, i \in L_e})$, and then consider:
\begin{align*}
\widehat{f_e}(s) &= \frac{1}{2^n}\sum_{z \in \lbrace 0, 1\rbrace^n} f_e(z)(-1)^{z \cdot s}\\
&\propto \sum_{z \in \lbrace 0, 1\rbrace^n} \mathbbm{1}\lbrack m_e(z\rho_e\pi_e, (\pi_v(\pi_e^{-1}(i)))_{v \in e, i \in L_e}) = m_e(x_e\rho_e\pi_e, (\pi_v(\pi_e^{-1}(i)))_{v \in e, i \in L_e}) (-1)^{z\cdot s}\rbrack\\
&= \sum_{z \in \lbrace 0, 1\rbrace^n} \mathbbm{1}\lbrack m_e(z\rho_e\pi_e, (\pi_v(\pi_e^{-1}(i)))_{v \in e, i \in L_e}) = m_e(x_e\rho_e\pi_e, (\pi_v(\pi_e^{-1}(i)))_{v \in e, i \in L_e}) (-1)^{z\rho_e\pi_e\cdot s\rho_e\pi_e}\rbrack\\
&= \sum_{z \in \{0,1\}^{L_e}} \mathbbm{1}\lbrack m_e(z, (\pi_v(\pi_e^{-1}(i)))_{v \in e, i \in L_e}) = m_e(x_e\rho_e\pi_e, (\pi_v(\pi_e^{-1}(i)))_{v \in e, i \in L_e}) \rbrack(-1)^{z \cdot s\rho_e\pi_e}
\end{align*}
which, as $L_e$ is fixed, is a deterministic function of $(s\rho_e\pi_e, x_e\rho_e\pi_e,  (\pi_v(\pi_e^{-1}(i)))_{v \in e, i \in L_e})$.

Now, for any $X$ and $\Pi$ that satisfy $\mathcal{E}$:
\begin{align}
||p(X, \Pi) - \nu||_{TV} &=  \frac{2^{|E|n - \varepsilon T}}{|J|}\left|\sum_{\substack{t \in \lbrace 0, 1\rbrace^{\eps T}\\ |t| \equiv 1 \bmod 2}}\prod_{e \in E}\widehat{f_e}(t0^{n-\varepsilon T})\notag\right|\\
&= \frac{\prod_{e \in E}| F_e|}{2^{\varepsilon T}|J|} \left|\sum_{\substack{t \in \lbrace 0, 1\rbrace^{\eps T}\\ |t| \equiv 1 \bmod 2}}\prod_{e \in E}\wt{f_e}(t0^{n-\eps T})\notag\right|\\
& \le \frac{3}{\delta}  \left|\sum_{\substack{t \in \lbrace 0, 1\rbrace^{\eps T}\\ |t| \equiv 1 \bmod 2}}\prod_{e \in E}\wt{f_e}(t 0^{n-\eps T})\right| \label{eq:ptv}
\end{align}

\noindent
For any $k \in [n]$, note that the distribution of a single
\[
  \wt{f_e}(s) = g_e(s\rho_e\pi_e, x_e\rho_e\pi_e, (\pi_v^{-1}\pi_e)_{v \in e})
\]
is identical for every $s \in \{0, 1\}^n$ of Hamming weight $k$:
$\pi_e$ permutes the first argument, and $x_e$ and $\pi_v$
independently permute the other ones given $\pi_e$.  Therefore for any
fixed $s \in \{0, 1\}^{n}$ of Hamming weight $k$, we have:
\begin{align*}
\beta_k := \max_e \E[\substack{X, \Pi}]{\wt{f_e}(s)^2} &= \max_e \frac{1}{{n \choose k}} \E[X, \Pi]{\sum_{\substack{s' \in \lbrace 0, 1\rbrace^n\\ |s'| = k}} \wt{f_e}(s')^2}
\end{align*}
is independent of which such $s$ is chosen.

\noindent
Because any function $f$ and all $s$ have $|\widehat{f}(s)| \le \E[y \sim \mathcal{U}(\lbrace 0, 1\rbrace^n)]{\abs{f(y)}}$, we also have: \[
  \abs{\wt{f_e}(s)} \in \lbrack 0, 1\rbrack.
\]
Therefore Lemma~\ref{lem:weights} with $q=2$ says for any $s$ with $\abs{s} = k$ that
\[
  \E[X,\Pi]{ \prod_{e \in E}\wt{f_e}(s)} \leq C \beta_k^\mu
\]
for some constant $C$ depending on the hypergraph $H$.  This lets us bound the expectation of
\begin{align*}
\sigma_k &= \left|\sum_{\substack{t \in \lbrace 0, 1\rbrace^{\varepsilon T}\\ |t| = k}}\prod_{e \in E}\wt{f_e}(t0^{n - \varepsilon T})\right|
\end{align*}
by
\[
  \E[X,\Pi]{\sigma_k} \leq C \binom{\eps T}{k} \beta_k^\mu.
\]
Our goal now is to bound the sum of this over all $1 \leq k \leq \eps T$.

\begin{description}
\item[Low-weight terms:] For $k \leq c$, by Lemma~\ref{lem:kklcor} we have
  \[
    \beta_{k} \leq \frac{1}{\binom{n}{k}} \left(\frac{2c}{k}\right)^{k}.
  \]
  Therefore
  \begin{align*}
    \sum_{k=1}^c \E[X,\Pi]{\sigma_k}
    &\leq C\sum_{k=1}^c \binom{\eps T}{k} \binom{n}{k}^{-\mu} \left(\frac{2c}{k}\right)^{k \mu}\\
    &\leq C\sum_{k=1}^c \left(\frac{2^\mu e \eps T c^\mu}{k n^\mu}\right)^{k}\\
    &\leq \frac{1}{20}\delta^2
  \end{align*}
  as long as $c \leq \gamma  n(\frac{ \delta^2}{\eps T})^{1/\mu}$ for a sufficiently small constant $\gamma$.
\item[High-weight terms:] By Parseval's identity,
  \begin{align*}
    \sum_{s \in \lbrace 0, 1\rbrace^n} \widehat{f_e}(s)^2 &= \frac{1}{2^n}\sum_{z \in \lbrace 0, 1\rbrace^n} f_e(z)^2
    = \frac{|F_e|}{2^n}
  \end{align*}
  so
  \begin{align*}
    \sum_{s \in \lbrace 0, 1\rbrace^n}\wt{f_e}(s)^2 \leq 2^c.
  \end{align*}
  and hence
  \[
    \sum_{k=0}^n \binom{n}{k} \beta_k \leq \abs{E} \cdot 2^c.
  \]
  Therefore, since $\mu \geq 1$,
  \begin{align*}
    \sum_{k=c+1}^{\varepsilon T} \E[X,\Pi]{\sigma_k}
    &\leq C\sum_{k=c+1}^{\varepsilon T} \binom{\eps T}{k} \beta_k^\mu\\
    &\leq C \abs{E} \cdot 2^c \max_{k: c \leq k \leq \varepsilon T } \frac{\binom{\eps T}{k}}{\binom{n}{k}}\\
    &\leq C \abs{E} \left(\frac{2e \eps T}{n}\right)^c.
  \end{align*}
  Since $\eps T \leq n/10$ and $c \geq 2\log(1/\delta) + C'$ for a chosen constant
  $C'$, we may choose $C'$ to be large enough that this gives
  \[
    \sum_{k=c+1}^{\eps T} \E[X,\Pi]{\sigma_k} \leq \delta^2/20.
  \]
\end{description}
Combining the two cases, we have
\[
\E[X, \Pi]{\sum_{k = 1}^{\eps T} \sigma_k} \le \delta^2/9
\]
and recall from~\eqref{eq:ptv} and the definition of $\sigma_k$ that
\[
\E[X, \Pi | \mathcal{E}]{||p(X,\Pi) - \nu||_{TV}} \le \frac{3}{\delta} \E[X, \Pi | \mathcal{E}]{\sum_{k = 1}^{\eps T} \sigma_k}.
\]
Since $||p(X,\Pi) - \nu||_{TV} \le 1$ always, this gives:
\begin{align*}
\E[X, \Pi]{||p(X,\Pi) - \nu||_{TV}} & \le \frac{3}{\delta}\E[X, \Pi]{\sum_{k = 1}^{\eps T} \sigma_k} + \Pb{\overline{\mathcal{E}}}\\
& \le \delta/3 + 2\delta/3\\
&= \delta.
\end{align*}
finishing the proof of Theorem~\ref{thm:lb-promise}.

\begin{corollary}\label{cor:lb-promise}
  Let $H$ be a connected hypergraph with more than one edge. Let
  $c \in \lbrack n \rbrack$. Suppose that, for all inputs $(X, \Pi)$
  to $\promise(H, n, T, \varepsilon)$, no player sends a message of more than $c$ bits.

Let $\mu = MVC_{1/2}(H)$, and let $0 < \delta < 1$.

There exists a constant $\gamma$, depending only on $H$, such that, if $c \le \gamma \frac{n}{(\delta^2 \varepsilon T)^{1/\mu}}$, the players succeed at the game with probability at most $1/2 + \delta$.
\end{corollary}
\begin{proof}
By Yao's principle \cite{Y77}, as we have a fixed distribution on inputs to our game, it is sufficient to consider deterministic protocols. Suppose we have such a protocol with maximum message size no more than $c$.

By Theorem \ref{thm:lb-promise}, the referee's posterior distribution on $\tau$ is at most $\delta$ from uniform after receiving the messages associated with the protocol, and therefore whatever function of the messages is used to guess $\tau$, it will be correct with probability at most $1/2 + \delta$.
\end{proof}

\section{Hypergraph Counting with No Promise}\label{sec:nopromisegame}
\subsection{Game}
We will define a $|E| + 1$-player game $\mathtt{Counting}(H,n,T,\varepsilon)$ (with $H$ a hypergraph, $n, T \in \mathbb{N}, \varepsilon \in (0,1))$ , as follows: There is one referee, who receives messages from every other player. No other communication takes place. Each player besides the referee corresponds to an edge $e \in E$. 

Let $N = T + (n - T)|E|$. For each edge $e \in E$, let $L_e \subset \lbrack N \rbrack$ be an $n$-element set containing $\lbrack T \rbrack$ and $n - T$ elements disjoint from every other $L_e$, so that if $a \not = b$, $L_a \cap L_b = \lbrack T \rbrack$, and $\bigcup_{e \in E} L_e = \lbrack N \rbrack$. For each $e \in E$, let $\rho_e : \lbrack n \rbrack \rightarrow L_e$ be a fixed bijection such that $\rho_e \vert_{\lbrack T \rbrack}$ is the identity.

An instance of $\mathtt{Counting}(H,n,T,\varepsilon)$ is as follows:

\begin{itemize}
\item For each edge $e \in E$:
\begin{itemize}
\item A string $x_e \in \lbrace 0, 1\rbrace^n$.
\item A string $\wt{x}_e$ generated by, for each bit of $x_e$, flipping that bit with probability $1/2 - \varepsilon^{1/|E|}/2$.
\item A permutation $\pi_e$ on $L_e$.
\end{itemize}
\item For each vertex $v \in V$:
\begin{itemize}
\item A permutation $\pi_v$ on $\lbrack N \rbrack$.
\end{itemize}
\item A string $\tau \in \lbrace 0^{T}, 1^{T}\rbrace$
\end{itemize}

\noindent
We will write $X$ for the strings $(x_e)_{e \in E}$, $\wt{X}$ for the strings $(\wt{x}_e)_{e \in E}$, $\chi$ for $(X, \wt{X})$, $\Pi_E$ for the permutations $(\pi_e)_{e \in E}$, and $\Pi_V$ for the permutations $(\pi_v)_{v \in V}$.
They have access to the following information:
\begin{itemize}
\item For each player $e \in E$:
\begin{itemize}
\item $\wt{x}_e \rho_e \pi_e$
\item $(\pi_v(\pi_e^{-1}(i)))_{v \in e, i \in L_e}$
\end{itemize}
\item For the referee:
\begin{itemize}
\item $\Pi = (\Pi_E, \Pi_V)$
\item $\tau \oplus \bigoplus_{e \in E}x_e^{1:T}$
\end{itemize}
\end{itemize}

\noindent
Given the messages received from the players, the referee's task will be to determine whether $\tau = 0^{T}$ or $\tau = 1^{T}$. 

\subsection{Hard Instance}
We will lower bound the complexity of this problem under the following hard instance: $\tau$ is chosen uniformly from $\lbrace 0^{ \varepsilon T}, 1^{\varepsilon T}\rbrace$, the strings $(x_e)_{e \in E}$ are each chosen uniformly and independently from $\lbrace 0, 1\rbrace^n$, and every permutation is chosen uniformly at random and independently of each other and the strings.

\subsection{Lower Bound}
For each player $e$, write $m_e(\wt{x}_e, (\pi_v(\pi_e^{-1}(i)))_{v \in e, i \in L_e})$ for the message the player sends to the referee on seeing $\wt{x}_e$ and $(\pi_v^{-1}\pi_e)_{v \in e}$.

\begin{theorem}
\label{thm:nopromisegame}
Let $H$ be a connected hypergraph with more than one edge. Let $c \in \lbrack n \rbrack$. Suppose that, for all inputs $(\chi, \Pi)$ to the game, no player sends a message of more than $c$ bits, and suppose $T < n/10$.

Let $p : \lbrace 0, 1 \rbrace^n \rightarrow \lbrack 0, 1\rbrack$ be the referee's posterior distribution on $\bigoplus_{e \in E}x_e^{1:T}$ \emph{before} considering $\tau \oplus \bigoplus_{e \in E}x_e^{1:T}$. Let $\upsilon$ be the distribution of $\mathcal{U}(\lbrace 0, 1\rbrace^T)$, the uniform distribution on the set $\lbrace 0, 1\rbrace^T$. Let $\mu = MVC_{1}(H)$.

There exists a constant $\gamma$ that depends on $H$ such that, if $c \le \gamma \frac{n}{(\delta^2 \varepsilon^2 T)^{1/\mu}}$:
\[
\E[\chi, \Pi]{||p - \upsilon||_{TV}} \le \delta
\]
\end{theorem}
\noindent
We will prove a weaker form of the theorem in which no player sends a
message of more than $c - 2 \log(3\abs{E}/\delta) - C'$ bits, for a
sufficiently large constant $C'$.  This implies the lemma statement
for a slightly larger $\gamma$, since the adjustment is
$O(\frac{n}{(\delta^2\varepsilon T)^{1/\mu}})$.

We will prove this by examining the Fourier coefficients of $p$.  We define the functions $(f_e)_{e \in E}, f$ and the sets $(F_e)_{e \in E}$ in a similar manner to the previous game.

For each $e$, define the random function $f_e : \lbrace 0, 1\rbrace^n \rightarrow \lbrace 0, 1 \rbrace$ by: \[
f_e(z) = \begin{cases}
1 & \mbox{if $m_e(z \rho_e \pi_e, (\pi_v(\pi_e^{-1}(i)))_{v \in e, i \in L_e}) = m_e(\wt{x}_e \rho_e \pi_e, (\pi_v(\pi_e^{-1}(i)))_{v \in e, i \in L_e})$.}\\
0 & \mbox{otherwise.}
\end{cases}
\]
And, with $Y = ((y_e)_{e\in E})$, define $f : \lbrace 0, 1\rbrace^{|E|n} \rightarrow \lbrace 0, 1 \rbrace$ by: \[
f(Y) = \prod_{e \in E}f_e(y_e)
\]
Let $F_e = f_e^{-1}(\lbrace 1\rbrace)$, and let: \[
F = \prod_{e \in E} F_e
\]

So then, as the game's inputs are uniformly distributed, the referee's posterior distribution on $\wt{X}$ is uniform on $F$, and so the posterior probability that $\wt{X} = \wt{Y}$ is: \[
\frac{f(\wt{Y})}{|F|}
\]
Therefore, the posterior probability that $X = Y$ is given by the sum over all $\wt{Y}$ of the probability that $\wt{X} = \wt{Y}$ times the probability of obtaining $Y$ from $\wt{Y}$ by flipping bits with probability $1/2 - \varepsilon^{1/|E|}/2$; that is:\[
\frac{\mathcal{T}_{\varepsilon^{1/|E|}}(f)(Y)}{|F|}
\]
Therefore, writing $Z = (z_e)_{e \in E}$, the referee's posterior distribution on $\bigoplus_{e \in E}x_e^{1:T}$ is given by:
\[
p(y) = \frac{1}{|F|} \sum_{\substack{z \in \lbrace 0, 1\rbrace^{|E|n}\\\bigoplus_{e \in E} z_e^{1:T} = y}} \mathcal{T}_{\varepsilon^{1/|E|}}(f)(z)
\]
So the Fourier coefficients of $p$ are given by:
\begin{align*}
\widehat{p}(s) &= \frac{1}{2^T}\sum_{y \in \lbrace 0, 1\rbrace^T} p(y) (-1)^{y \cdot s}\\
&=   \frac{1}{2^T|F|}\sum_{y \in \lbrace 0, 1\rbrace^T}\sum_{\substack{z \in \lbrace 0, 1\rbrace^{|E|n}\\\bigoplus_{e \in E} z_e^{1:T} = y}} \mathcal{T}_{\varepsilon^{1/|E|}}(f)(z) (-1)^{y \cdot s}\\ 
&= \frac{1}{2^T|F|}\left(\sum_{\substack{z \in \lbrace 0, 1\rbrace^{|E|n}\\ \bigoplus_{e \in E} z_e^{1:T}\cdot s = 0}} \mathcal{T}_{\varepsilon^{1/|E|}}(f)(z) - \sum_{\substack{z \in \lbrace 0, 1\rbrace^{|E|n}\\ \bigoplus_{e \in E} z_e^{1:T}\cdot s = 1}} \mathcal{T}_{\varepsilon^{1/|E|}}(f)(z)\right)\\
&= \frac{1}{2^T|F|}\left(\sum_{\substack{z \in \lbrace 0, 1\rbrace^{|E|n}\\ z \cdot (s0^{n-T})_{e\in E} = 0}} \mathcal{T}_{\varepsilon^{1/|E|}}(f)(z) - \sum_{\substack{z \in \lbrace 0, 1\rbrace^{|E|n}\\ z \cdot (s0^{n-T})_{e\in E} = 1}} \mathcal{T}_{\varepsilon^{1/|E|}}(f)(z)\right)\\
&= \frac{2^{|E|n - T}}{|F|} \widehat{\mathcal{T}_{\varepsilon^{1/|E|}}(f)}((s0^{n-T})_{e\in E})\\
&= \frac{2^{|E|n - T}}{|F|}  \varepsilon^{|s|} \widehat{f}((s0^{n-T})_{e\in E})
\end{align*}

\noindent
And so by applying Lemma \ref{lem:fdecomp}: \[
\widehat{p}(s) = \frac{2^{|E|n - T}}{|F|}  \varepsilon^{|s|} \prod_{e \in E}  \widehat{f}(s0^{n-T}).
\]
By Parseval's identity, as for any probability
distribution $q$ on $\lbrace 0, 1\rbrace^{T}$,
$\widehat{q}(0^T) = \frac{1}{2^T}$ and $\wh{v}(s) = 0$ for all
$s \neq 0^T$:
\begin{align*}
\sum_{z \in \lbrace 0, 1\rbrace^T} (p(z) - \upsilon(z))^2 &= \frac{2^{2|E|n - T}}{|F|^2}  \sum_{s \in \lbrace 0, 1\rbrace^T \setminus \lbrace 0^T\rbrace} \varepsilon^{2|s|}\prod_{e \in E}  \widehat{f}(s0^{n-T})^2\\
&=  2^{-T}  \sum_{s \in \lbrace 0, 1\rbrace^T \setminus \lbrace 0^T\rbrace} \varepsilon^{2|s|}\prod_{e \in E}  \frac{2^{2n}}{|F_e|^2}\widehat{f}(s0^{n-T})^2
\end{align*}

\noindent
By applying Lemma \ref{lem:setsizes} with $\alpha = \log 2|E|/\delta$ and applying the union bound: \[
\forall e \in E, |F_e| \le 2^{n -c}
\]
conditioned on an event $\mathcal{E}$ with probability at least $1 - \delta/2$.

We will now define a renormalized and masked version of $\wh{f}_e$ as follows:
\[
\wt{f}_e(s) := \begin{cases}
\frac{2^{n}\widehat{f_e}(s)}{|F_e|} & \mbox{if $|F_e| \ge 2^{n-c}$}\\
0 & \mbox{otherwise.}
\end{cases}
\]
Note that $\wt{f}_e$ can be expressed as a \emph{deterministic} function of the randomness in the following form:
\[
  \wt{f}_e(s) = g_e(s\rho_e\pi_e, \wt{x}_e\rho_e\pi_e, (\pi_v(\pi_e^{-1}(i)))_{v \in e, i \in L_e})
\]
To justify this, first recall that $F_e$ is determined by $(\wt{x}_e\rho_e\pi_e, \wt{x}_e\rho_e\pi_e, (\pi_v(\pi_e^{-1}(i)))_{v \in e, i \in L_e}$, and then consider: 
\begin{align*}
\widehat{f_e}(s) &= \frac{1}{2^n}\sum_{z \in \lbrace 0, 1\rbrace^n} f_e(z)(-1)^{z \cdot s}\\
&\propto \sum_{z \in \lbrace 0, 1\rbrace^n} \mathbbm{1}\lbrack m_e(z\rho_e\pi_e, (\pi_v(\pi_e^{-1}(i)))_{v \in e, i \in L_e}) = m_e(\wt{x}_e\rho_e\pi_e, (\pi_v(\pi_e^{-1}(i)))_{v \in e, i \in L_e}) (-1)^{z\cdot s}\rbrack\\
&= \sum_{z \in \lbrace 0, 1\rbrace^n} \mathbbm{1}\lbrack m_e(z\rho_e\pi_e, (\pi_v(\pi_e^{-1}(i)))_{v \in e, i \in L_e}) = m_e(\wt{x}_e\rho_e\pi_e, (\pi_v(\pi_e^{-1}(i)))_{v \in e, i \in L_e}) (-1)^{z\rho_e\pi_e\cdot s\rho_e\pi_e}\rbrack\\
&= \sum_{z \in L_e} \mathbbm{1}\lbrack m_e(z, (\pi_v(\pi_e^{-1}(i)))_{v \in e, i \in L_e}) = m_e(\wt{x}_e\rho_e\pi_e, (\pi_v(\pi_e^{-1}(i)))_{v \in e, i \in L_e}) \rbrack(-1)^{z \cdot s\rho_e\pi_e}
\end{align*}
As $L_e$ is fixed, this is a deterministic function of $(s\rho_e\pi_e, (\wt{x}_e\rho_e\pi_e,  (\pi_v(\pi_e^{-1}(i)))_{v \in e, i \in L_e}))$.

Now, conditioned on $\mathcal{E}$:
\begin{align}
\sum_{z \in \lbrace 0, 1\rbrace^T} (p(z) - \upsilon(z))^2 &=  2^{-T} \sum_{s \in \lbrace 0, 1\rbrace^T \setminus \lbrace 0^T\rbrace} \eps^{2\abs{s}}\prod_{e \in E}  \wt{f}(s0^{n-T})^2.\label{eq:l2sum}
\end{align}

\noindent
For any $k \in [n]$, note that the distribution of a single
\[
  \wt{f_e}(s) = g_e(s\rho_e\pi_e, \wt{x}_e\rho_e\pi_e, (\pi_v^{-1}\pi_e)_{v \in e})
\]
is identical for every $s \in \{0, 1\}^n$ of Hamming weight $k$:
$\pi_e$ permutes the first argument, and $x_e$ and $\pi_v$
independently permute the other ones given $\pi_e$.  Therefore for any
fixed $s \in \{0, 1\}^{n}$ of Hamming weight $k$, we have:
\begin{align*}
\beta_k := \max_e \E[\substack{X, \Pi}]{\wt{f_e}(s)^2} &= \max_e \frac{1}{{n \choose k}} \E[X, \Pi]{\sum_{\substack{s' \in \lbrace 0, 1\rbrace^n\\ |s'| = k}} \wt{f_e}(s')^2}
\end{align*}
independent of which such $s$ is chosen.

\noindent
Because any function $f$ and all $s$ have $|\widehat{f}(s)| \le \E[y \sim \mathcal{U}(\lbrace 0, 1\rbrace^n)]{\abs{f(y)}}$, we also have: \[
  \wt{f_e}(s)^2 \in \lbrack 0, 1\rbrack.
\]
Therefore Lemma~\ref{lem:weights} with $q=1$ says for any $s$ with $\abs{s} = k$ that
\[
  \E[X,\Pi]{ \prod_{e \in E}\wt{f_e}(s)^2} \leq C \beta_k^\mu
\]
for some constant $C$ depending on the hypergraph $H$.  This lets us bound the
expectation of
\begin{align*}
\sigma_k &= \sum_{\substack{t \in \lbrace 0, 1\rbrace^{T}\\ |t| = k}} \eps^{2k}\prod_{e \in E}\wt{f_e}(t0^{n - T})^2
\end{align*}
by
\[
  \E[X,\Pi]{\sigma_k} \leq C \eps^{2k}\binom{T}{k}\beta_k^\mu.
\]
Our goal now is to bound the sum of this over all $1 \leq k \leq T$.
\begin{description}
\item[Low-weight terms:] For $k \leq c$, by Lemma~\ref{lem:kklcor} we have
  \[
    \beta_{k} \leq \frac{1}{\binom{n}{k}} \left(\frac{2c}{k}\right)^{k}.
  \]
  Therefore
  \begin{align*}
    \sum_{k=1}^c \E[X,\Pi]{\sigma_k}
    &\leq C\sum_{k=1}^c \eps^{2k}\binom{T}{k} \binom{n}{k}^{-\mu} \left(\frac{2c}{k}\right)^{k \mu}\\
    &\leq C\sum_{k=1}^c \left(\frac{2^\mu e \eps^2 T c^\mu}{k n^\mu}\right)^{k}\\
    &\leq \frac{1}{20}\delta^2
  \end{align*}
  as long as $c \leq \gamma  n(\frac{ \delta^2}{\eps^2 T})^{1/\mu}$ for a sufficiently small constant $\gamma$.
\item[High-weight terms:] By Parseval's identity,
  \begin{align*}
    \sum_{s \in \lbrace 0, 1\rbrace^n} \widehat{f_e}(s)^2 &= \frac{1}{2^n}\sum_{z \in \lbrace 0, 1\rbrace^n} f_e(z)^2
    = \frac{|F_e|}{2^n}
  \end{align*}
  so
  \begin{align*}
    \sum_{s \in \lbrace 0, 1\rbrace^n}\wt{f_e}(s)^2 \leq 2^c
  \end{align*}
  and hence
  \[
    \sum_{k=0}^n \binom{n}{k} \beta_k \leq \abs{E} \cdot 2^c.
  \]
  Therefore, since $\mu \geq 1$ and $\eps^2T \leq n/10$,
  \begin{align*}
    \sum_{k=c+1}^{T} \E[X,\Pi]{\sigma_k}
    &\leq C\sum_{k=c+1}^{T} \eps^{2k}\binom{T}{k}\beta_k^\mu\\
    &\leq C \abs{E} \cdot 2^c \max_{k: c \leq k \leq T } \frac{\eps^{2k}\binom{T}{k}}{\binom{n}{k}}\\
    &\leq C \abs{E} \left(\frac{2e \eps^2 T}{n}\right)^c\\
    &\leq C \abs{E} 2^{-c}.
  \end{align*}
  Since $c \geq 2\log(1/\delta) + C'$ for a sufficiently large constant
  $C'$, this gives
  \[
    \sum_{k=c+1}^{T} \E[X,\Pi]{\sigma_k} \leq \delta^2/20.
  \]
\end{description}
Combining the two cases, we have
\[
\E[X, \Pi]{\sum_{k = 1}^{T} \sigma_k} \le \delta^2/9.
\]
Then: \[
\E[X, \Pi | \mathcal{E}]{ \sum_{z \in \lbrace 0, 1\rbrace^T} (p(z) - \upsilon(z))^2 } \le 2^{-T}\E[X, \Pi | \mathcal{E}]{\sum_{k = 1}^{T} \sigma_k}
\]
So: \[
\E[X, \Pi | \mathcal{E}]{|| p - \upsilon||_{TV}^2} \le \E[X, \Pi | \mathcal{E}]{\sum_{k = 1}^{T} \sigma_k}
\]
And so, as $|| p - \upsilon||_{TV}^2 \le 1$ always: 
\begin{align*}
  \E[X, \Pi]{|| p - \upsilon||_{TV}}
  &\leq \E[X, \Pi]{|| p - \upsilon||_{TV} \mathbbm{1}_{\mathcal{E}}} + \E[X, \Pi]{\mathbbm{1}_{\overline{\mathcal{E}}}}\\
  &\leq \E[X, \Pi]{\sum_{k=1}^T \sigma_k}^{1/2} + \Pb{\overline{\mathcal{E}}}\\
  &\leq \delta/3 + \delta/2\\
&= \delta
\end{align*}
as desired.
\begin{corollary}\label{cor:nopromisegame}
  Let $H$ be a connected hypergraph with more than one edge. Let
  $c \in \lbrack n \rbrack$. Suppose that, for all inputs $(X, \Pi)$
  to $\nopromise(H, n, T, \varepsilon)$, no player sends a message of more than $c$ bits.

Let $\mu = MVC_{1}(H)$, and let $0 < \delta < 1$.

There exists a constant $\gamma$ that depends on $H$ such that, if $c \le \gamma \frac{n}{(\delta^2 \varepsilon^2 T)^{1/\mu}}$, the players succeed at the game with probability at most $1/2 + \delta$.\end{corollary}
\begin{proof}
By Yao's principle \cite{Y77}, as we have a fixed distribution on inputs to our game, it is sufficient to consider deterministic protocols. Suppose we have such a protocol with maximum message size
no more than $c$.

By applying Theorem \ref{thm:nopromisegame} with a smaller choice of constant $\gamma$, the referee's posterior distribution on $\bigoplus_{e \in E}x_e^{1:T}$ is at most $\delta/2$ from uniform after receiving the messages associated with the protocol but before looking at $\tau \oplus \bigoplus_{e \in E}x_e^{1:T}$. So then, after looking at $\tau \oplus \bigoplus_{e \in E}x_e^{1:T}$, the referee must determine whether it is more likely that they are looking at $ \bigoplus_{e \in E}x_e^{1:T}$ or $1^T \oplus \bigoplus_{e \in E}x_e^{1:T}$. However, the distributions of  $ \bigoplus_{e \in E}x_e^{1:T}$ and $1^T \oplus \bigoplus_{e \in E}x_e^{1:T}$ conditioned on the messages received are both $\delta/2$-close to uniform, and so by the triangle inequality are at most $\delta$ from each other, and so the referee guesses correctly with probability at most $1/2 + \delta$.
\end{proof}

\section{Linear Sketching Lower Bound}\label{sec:reduction}
\begin{definition}
\label{dfn:composable}
Let $\mathcal{A}$ be a randomized graph streaming algorithm, and let $\mathbb{S}$ be the set of possible states $\mathcal{S}$ of $\mathcal{A}$. We will say $\mathcal{A}$ has \emph{composable state} if, for any fixed random seed for $\mathcal{A}$, there is a function $c: \mathbb{S} \times \mathbb{S} \rightarrow \mathbb{S}$ such that, if $\mathcal{S}_1$ is the state of $\mathcal{A}$ after receiving the stream of edges $E_1$ as input, and $\mathcal{S}_2$ is the state of $\mathcal{A}$ after receiving the stream of edges $E_2$ as input, $c(\mathcal{S}_1,\mathcal{S}_2)$ is the state of $\mathcal{A}$a after receiving the concatenation of $E_1$ and $E_2$ as input.
\end{definition}

\begin{theorem}
\label{thm:graphsketching}
Let $H = (V, E)$ be a (fixed) connected hypergraph with $|E| > 1$. Let $T \in \mathbb{N}, \varepsilon \in (1/\sqrt{T}, 1\rbrack$. Let $\mathcal{A}$ be a graph streaming algorithm that can distinguish between graphs $G$ presented as a stream of edges with at least $T$ copies of $H$ and graphs with at most $(1 - \varepsilon)T$ copies of $H$ with probability $99/100$, provided $G$ has no more than $m$ edges. Let $S(m)$ be the maximum space usage of $\mathcal{A}$ across all $m$-edge inputs.

Furthermore, let $\mathcal{A}$ have composable state. Then, for all $m \ge O(T)$: \[
S(m) = \Omega\left(\max \left(\frac{m}{(\varepsilon T)^{1/\mu_2}}, \frac{m}{(\varepsilon^2 T)^{1/\mu_1}}\right)\right)
\]
where $\mu_2 = MVC_\frac{1}{2}(H)$ and $\mu_1 = \max_{e \in E}MVC_1(H \setminus e)$, with constants that may depend on $H$ but nothing else.
\end{theorem}

\noindent
We will prove this by reductions to $\mathtt{PromiseCounting}$ and $\mathtt{Counting}$. In both reductions, we will use the following lemma on binomial distributions, from \cite{KB80}: 
\begin{lemma}
Let $m$ be any median of $\text{Bi}(n,p)$. Then: \[
\lfloor np \rfloor \le m \le \lceil np \rceil
\]
\end{lemma}
\noindent
Our first reduction will be to $\promise$.

\begin{lemma}
\[
\forall m \ge O(T), S(m) = \Omega\left(\frac{m}{(\varepsilon T)^{1/\mu_2}}\right)
\]
\end{lemma}
\begin{proof}
  First, we note that we may assume that $T \geq 10000$ WLOG.  An
  algorithm for $T < 10000$ can distinguish between streams with $0$
  copies of $H$ and streams with at least $20000$ copies of $H$, and
  applying the lemma for $T=10000$ and $\eps=1$ will get the desired
  $\Omega(m)$ bound.  We may assume $10/\sqrt{T} \le \eps \leq 1/10$ for similar
  reasons.  We will also assume $\varepsilon$ is an integer multiple
  of $1/T$, as this will cost us at most a factor of $2$ in the bound,
  as if $\mathcal{A}$ can distinguish between graphs with
  $(1 - \varepsilon)T$ and $T$ copies of $H$, it can
  distinguish between graphs with $1 - \lfloor \varepsilon T\rfloor$
  and $T$ copies.

 We will use $\mathcal{A}$ to devise a $S(m)$-bit protocol for $\promise(H,n,T',\varepsilon)$, where $n = \Theta(m)$, $T' = 2^{|E|}T$, and the instance is distributed as in our ``hard instance'' from Theorem \ref{thm:lb-promise}. This will be based on constructing a graph $G = (B,R)$ in pieces to input to $\mathcal{A}$.  The protocol is as follows (recall that $\bigcup_{e \in E} L_e = \lbrack N\rbrack$): 
\begin{itemize}
\item Let $B = \lbrack N\rbrack \times V(H)$
\item Each player $e$, on seeing the input:
\begin{itemize}
\item $x_e\rho_e\pi_e \in \lbrace 0, 1\rbrace^{L_e}$
\item $(\pi_v(\pi_e^{-1}(i)))_{v \in e, i \in L_e}$
\end{itemize}
constructs a set of edges $R_e$ as follows: For each $i \in L_e$ such that $(x_e\rho_e\pi_e)_i = 0$, add the hyperedge $\lbrace (\pi_v(\pi_e^{-1}(i))), v) : v \in e \rbrace$.
\item Each player $e$ runs $\mathcal{A}$ with $R_e$ as input, and then sends the state of $\mathcal{A}$ to the referee.
\item The referee composes all received states, and reads the output of $\mathcal{A}$ on $G = (B, R)$, where $R = \bigcup_{e\in E} R_e$.
\item If the algorithm reports that $G$ has no more than $(1 - \varepsilon T)$ triangles, the referee decides that $\tau = 1^{\varepsilon T}$, and otherwise decides that $\tau = 0^{\varepsilon T}$.
\end{itemize}
Now we would like to know how many copies of $H$ are in $G = (B, R)$. Consider any vertex $(i, v) \in \lbrack N\rbrack \times V(H)$. If $\pi_v^{-1}(i) \not \in \lbrack T'\rbrack$, then at most one edge in $R$ includes $(i, v)$, as there is at exactly one $e \in E$ such that $L_e$ includes $i$ (as the sets $L_e$ were defined to have pairwise intersection $\lbrack T'\rbrack$).

For any $i$ such that $\pi_v^{-1} \in \lbrack T'\rbrack$, and for each $e \ni v$, $(i, v)$ will be contained in the hyperedge $\lbrace (\pi_v(\pi_e^{-1}(j))), v) : v \in e \rbrace$ (if it exists, that is if $(x_e\rho_e\pi_e)_j = 0$), where $j = \pi_e(\pi_v^{-1}(i))$, and no other edges. As this applies for each $v \in V(H)$, the connected component containing $(i, v)$ will be contained in $R^{(i)} = \lbrace \lbrace (\pi_v(\pi_e^{-1}(j))), v) : v \in e \rbrace : e \in E, j = \pi_e(\pi_v^{-1}(i))\rbrace$. Therefore, the other vertices in edges that contain $(i, v)$ will be contained within the set $\lbrace (\pi_u(\pi^{-1}_v(i)) , u) : \exists e \ni v, v \in e \rbrace$.

By repeating this argument, this means that the connected component containing $(i, v)$ will be contained in $\lbrace (\pi_u(\pi^{-1}_v(i)) , u) : v \in V(H)$. So this component contains exactly one copy of $H$ if the edge $\lbrace (\pi_v(\pi_u^{-1}(i)), v) : v \in e \rbrace$ is present for every $e \in H$, and no copies otherwise. The first will happen iff, for every $e \in E$, the $\pi_e(\pi_v^{-1}(i))$ bit of $x_e\rho_e\pi_e$ is 0, so if the $\pi_v^{-1}(i)$ bit of $x_e\rho_e$ is 0. As $\pi_v^{-1} \in \lbrack T'\rbrack$ and $\rho_e$ is the identity on $\lbrack T'\rbrack$, this happens iff the $\pi_v^{-1}(i)$ bit of $x_e$ is 0.

So, as any copy of $H$ must contain at least one vertex of degree at least 2, the number of copies of $H$ in $G$ is equal to the number of indices $i\in \lbrack T'\rbrack$ such that $\forall e, (x_e)_i = 0$. Recall that in our hard instance of $\mathtt{PromiseCounting}$ the strings $(x_e)_{e \in E}$ are uniform on strings such that: \[
\bigoplus_{e \in E} x_e^{1:\varepsilon T} = \tau
\]
where $\tau$ is uniformly distributed on $\lbrace 0^{\varepsilon T}, 1^{\varepsilon T}\rbrace$.

Now, for any $i$ such that $\bigoplus_{e \in E} (x_e)_i = 1$, there is at least one $e \in E$ such that $(x_e)_i = 1$, while if $\bigoplus_{e \in E} (x_e)_i = 0$, it is the case that $\forall e \in E, (x_e)_i = 0$ with probability $2^{1 - |E|}$. For $\varepsilon T < i \le T$, $\bigoplus_{e \in E} (x_e)_i$ is equally likely to be either, while for $i \le \varepsilon T$, $\bigoplus_{e \in E} (x_e)_i = \tau_i$.

Therefore, the number of copies of $H$ in $G$ is distributed as: 
\begin{align*}
&\text{Bi}((1 - \varepsilon)T', 2^{-|E|}) & \mbox{Conditioned on $\tau = 1^{\varepsilon T'}$.} \\
&\text{Bi}((1 - \varepsilon)T', 2^{-|E|}) + \text{Bi}(\varepsilon T', 2^{1 - |E|}) & \mbox{Conditioned on $\tau = 0^{\varepsilon T'}$.}
\end{align*}
As $\varepsilon$ is an integer multiple of $T$, the unique median of $\text{Bi}((1 - \varepsilon)T', 2^{-|E|})$ is: \[
(1 - \varepsilon)T
\]
Therefore, the probability that $\text{Bi}((1 - \varepsilon)T', 2^{-|E|}) \le (1 - \varepsilon)T$ is at least $1/2$.  Then, as $\var(\text{Bi}((1 - \varepsilon)T', 2^{-|E|}) + \text{Bi}(\varepsilon T', 2^{1 - |E|})) < (1 + \varepsilon)T'2^{-|E|} = (1 + \varepsilon)T$, by Chebyshev: \[
\Pb{\text{Bi}((1 - \varepsilon)T', 2^{-|E|}) + \text{Bi}(\varepsilon T', 2^{1 - |E|}) \le T} \le \frac{(1 + \varepsilon)^2}{4\varepsilon^2 T} \le \frac{1}{100}
\]
Therefore, when $\tau = 0^{\varepsilon T}$, the protocol correctly guesses it with probability at least $99/200$,and when $\tau = 1^{\varepsilon T}$, the protocol correctly guesses it with probability at least $(99/100)^2 \ge 98/100$. So the success probability of this protocol is at least: \[
\frac{1}{2} \cdot \frac{99}{200} + \frac{1}{2} \cdot \frac{98}{100} = 0.7375
\]

Then, by Corollary \ref{cor:lb-promise}, this implies that: \[
S(m) = \Omega\left(\frac{m}{(\varepsilon T)^{1/\mu_2}}\right)
\]
\end{proof}

\begin{lemma}
\[
\forall m \ge O(T), \forall e^* \in E, S(m) = \Omega\left(\frac{m}{(\varepsilon^2 T)^{1/\mu_1}}\right)
\]
where $\mu_1 = MVC_1(H \setminus e^*)$. 
\end{lemma}
\begin{proof}
As in the previous lemma, we will assume that $T \ge 100$, $10000/\sqrt{T} \le \varepsilon \le 1/10$, and $\varepsilon$ is an integer multiple of $1/T$, at the cost of at most a constant factor in our bound. 

Let $\mathcal{A}$ be a composable distinguishing algorithm. We will use it to devise a $S(m)$-bit protocol for $\mathtt{Counting(H \setminus e^*, n, T', \varepsilon)}$, where $T' = 2^{|E|}T$ and the state of the game is distributed according to our ``hard instance'' from Theorem \ref{thm:nopromisegame}. Let $E' = E \setminus e^*$.
\begin{itemize}
\item Let $B = \lbrack N\rbrack \times V(H)$
\item Each player $e$, on seeing the input:
\begin{itemize}
\item $\wt{x_e}\rho_e\pi_e \in \lbrace 0, 1\rbrace^{L_e}$
\item $(\pi_v(\pi_e^{-1}(i)))_{v \in e, i \in L_e}$
\end{itemize}
constructs a set of edges $R_e$ as follows: For each $i \in L_e$ such that $(\wt{x}_e\rho_e\pi_e)_i = 0$, add the hyperedge $\lbrace (\pi_v(\pi_e^{-1}(i))), v) : v \in e \rbrace$.
\item Each player runs $\mathcal{A}$ with $R_e$ as input, and then sends the state of $\mathcal{A}$ to the referee.
\item The referee, on seeing the input:
\begin{itemize}
\item $\Pi = (\Pi_E, \Pi_V)$
\item $\tau \oplus \bigoplus_{e \in E}x_e^{1:T}$
\end{itemize}
sets $\wt{x}_{e^*} = \tau \oplus \bigoplus_{e \in E}x_e^{1:T}$, chooses $\rho_{e^*}, \pi_{e^*}$ arbitrarily, and constructs $R_{e^*}$ by adding \\$\lbrace (\pi_v(\pi_{e^*}^{-1}(i))), v) : v \in e^* \rbrace$ for each $i \in \lbrack T\rbrack$ such that $(\wt{x}_e\rho_e\pi_e)_i = 1$.
\item The referee runs $\mathcal{A}$ with $R_{e^*}$ as input, and then composes the state of $\mathcal{A}$ with the received states, and reads off the output of $\mathcal{A}$.
\item If the algorithm reports that $G$ has no more than $(1 - \varepsilon T)$ triangles, the referee decides that $\tau = 1^{T}$, and otherwise decides that $\tau = 0^{T}$.
\end{itemize}
As in the previous lemma, let $G$ be $(B, R = \bigcup_{e \in E} R_e)$. By the same argument, the number of copies of $H$ in $G$ will be precisely the number of indices $i \in \lbrack T\rbrack$ such that $\forall e \in E, (\wt{x}_e)_i = 0$.

To analyze this, first recall that $\wt{x_e}$ was generated from $x_e$ by flipping every bit of $x_e$ independently with probability $1/2 - \varepsilon^{1/|E|}/2$, and so for each $e \in E'$, we can write $\wt{x}_e = x_e \oplus y_e$, where the $y_e$ are independent and are generated by setting each co-ordinate of $y_e$ independently to $1$ with probability $1/2 - \varepsilon^{1/|E|}/2$ and $0$ otherwise. Recall also that the strings $(x_e)_{e \in E}$ are uniformly distributed, and so conditioned on $\wt{x}_{e^*}$, they are are distributed uniformly among strings that sum to $\wt{x}_{e^*}$. Therefore, if we condition on $(y_e)_{e \in E'}$ and $\tau$ the $(\wt{x}_e)_{e \in E'}$ are distributed uniformly among strings such that: \[
\bigoplus_{e \in E'} \wt{x}_e = \tau \oplus \wt{x}_{e^*} \oplus \bigoplus_{e \in E'} y_e 
\]
Using the fact that the probability of $\text{Bi}(n, p)$ being even is $1/2 + (1 - 2p)^n/2$ (see, e.g.\ the proof in \cite{stackex:16187}), for each $i \in \lbrack T \rbrack$ we have that $\bigoplus_{e \in E'} (\wt{x}_e)_i = (\tau \oplus \wt{x}_{e^*})_i$ with probability $1/2 + \varepsilon/2$. Therefore, when $(\wt{x}_{e^*})_i = 0$, the probability that $\forall e \in E, (\wt{x}_e)_i = 0$ is: 
\begin{align*}
&(1 + \varepsilon)2^{1-|E|} & \mbox{If $\tau_i = 0$}\\
&(1 - \varepsilon)2^{1-|E|} & \mbox{If $\tau_i = 1$}
\end{align*}
While when $(\wt{x}_{e^*})_i = 1$, the probability is 0 by definition. So, as $\wt{x}_{e^*}$ is uniformly distributed when only conditioned on $\tau$, we can write down the distribution on the number of copies of $H$ in $G$: 
\begin{align*}
&\text{Bi}(T', (1 +\varepsilon)2^{-|E|}) & \mbox{If $\tau = 0^{T'}$}\\
&\text{Bi}(T',(1 - \varepsilon)2^{-|E|}) & \mbox{If $\tau = 1^{T'}$}
\end{align*}
So by considering the median, the probability that $\text{Bi}(T', (1 -\varepsilon)2^{-|E|}) \le (1-\varepsilon)T$ is at least $1/2$, while by Chebyshev the probability that $\text{Bi}(T',(1 + \varepsilon)2^{-|E|}) \le  (1-\varepsilon)T$ is at most $1/100$, and so as in the previous lemma, the protocol succeeds with probability at least $0.7375$.

Therefore, by Corollary \ref{cor:nopromisegame}: \[
 S(m) = \Omega\left(\frac{m}{(\varepsilon^2 T)^{1/\mu_1}}\right)
\]
\end{proof}
\noindent
Theorem \ref{thm:graphsketching} then follows directly from the previous two lemmas.

To prove this gives tight bounds (for $\varepsilon$ constant) for all $2$-uniform hypergraphs (that is, all graphs), we will need the following lemma on graph covers:
\begin{lemma}
\label{lem:hypographcovers}
Let $G = (V, E)$ be a connected graph with $|E| > 1$. Then: \[
\max(MVC_2(G), \max_{e \in E} MVC_1(G \setminus e)) = MVC_1(G)
\]
is the standard fractional vertex cover of $G$.
\end{lemma}
\begin{proof}
Consider the dual of the fractional vertex cover problem, the fractional maximum matching problem, where the aim is to find a function $f : E \rightarrow \lbrack 0, 1 \rbrack$ such that \[
\forall u \in V, \sum_{v \in N(u)} f(uv) \le 1
\]
and $\sum_{e \in E} f(e)$ is maximized. This is known (see, e.g.,~\cite{havet:lecturenotes}) to have a half-integral optimal solution, and therefore a solution: \[
f(e) = \begin{cases}
1 & \mbox{if $e \in D$}\\
1/2 & \mbox{if $e \in C$}
\end{cases}
\]
where $D$ is a (possibly empty) set of disjoint edges, and $C$ is either empty or an odd cycle disjoint from $D$. If $G \not= C$ then, as $G$ is connected, there is at least one edge $e$ such that $f(e) = 0$. Therefore, that edge can be deleted from $G$ without changing its maximum matching number and therefore without changing $MVC_1(G)$, and so $MVC_1(G) = \max_{e \in E} MVC_1(G \setminus e)$.

Otherwise, suppose $G$ is an odd cycle. Let $g : V \cup E \rightarrow \lbrack 0, \infty)$ be any function such that: \[
\sum_{v \in e} (g(v) + g(e)) \ge 1, \forall e \in E
\]
Then:
\begin{align*}
\sum_{v \in V} g(v) + \frac{1}{2}\sum_{e \in E} g(e) &\ge \sum_{v \in V} g(v) + \frac{1}{2}\sum_{uv \in E}( 1 - g(u) - g(v))\\
&= |E|/2
\end{align*}
So $MVC_\frac{1}{2}(G) = |E|/2$, which is also $MVC_1(G)$ (by considering a cover that puts weight $1/2$ on each vertex).
\end{proof}

\begin{corollary}
\label{cor:hypographlb}
Let $H = (V, E)$ be a connected graph with $|E| > 1$. Let $\varepsilon \in (0, 1\rbrack, T \in \mathbb{N}$. Let $\mathcal{A}$ be a graph streaming algorithm that can distinguish between graphs $G$ presented as a stream of edges with $T$ copies of $H$, graphs with $(1 - \varepsilon)T$ copies of $H$ with probability $99/100$, provided $G$ has no more than $m$ edges. Let $S(m)$ be the maximum space usage of $\mathcal{A}$ across all $m$-edge inputs.

Furthermore, let $\mathcal{A}$ have composable state. Then: \[
\forall m \ge O(T), S(m) = \Omega\left(\frac{m}{(\varepsilon T)^{1/\tau}}\right)
\]
where $\tau$ is the fractional vertex cover of $H$, and the constant factor may depend on $H$ but nothing else.
\end{corollary}

%!TEX root = ./main.tex
\section{Upper Bound}\label{sec:upperbound}

\newcommand{\expect}{{\mathbb{E}}}
\newcommand{\prob}{{\mathbb{P}}}
\newcommand{\Aut}{\text{Aut}}
\newenvironment{proofof}[1]{\noindent{\bf Proof of #1:}}{$\qed$\par}

Our main result is 
\begin{theorem}\label{thm:hypergraphs-ub}
  For every hypergraph $H=(V_H, E_H)$, $\e\in (0, 1)$ there exists a
  sketching algorithm that, for any hypergraph $G=(V_G, E_G)$ on $n$ vertices with degrees
  bounded by $d$, approximates the number of copies of
  $H$ in $G$ to within a $1+\e$ multiplicative factor with probability at least
  $99/100$ using space
  $s\leq C \cdot \e^{-2/\tau} \cdot m T^{-1/\tau}$, where $T$ is the
  number of copies of $H$ in $G$ and $\tau$ is the fractional
  vertex cover of $H$ and $C$ is a constant that depends on $H$.
\end{theorem}

\noindent
For graphs we get  a more powerful result, which allows the graph $G$ to have higher degrees:

\begin{theorem}\label{thm:graphs-ub}
  For every graph $H=(V_H, E_H)$ that admits a minimum vertex cover
  that assigns nonzero weight to every vertex, for every
  $\e\in (0, 1)$ there exists a sketching algorithm that, for any graph $G=(V_G, E_G)$ on $n$
  vertices with degrees bounded by
  $d\leq C' \e^{1/\tau} T^{1/(2\tau)}$, approximates
  the number of copies of $H$ in $G$
  to within a $1+\e$ multiplicative factor with probability $99/100$ using
  space $C \cdot \e^{-2/\tau} \cdot m T^{-1/\tau}$, where $T$ is the
  number of copies of $H$ in $G$ and $\tau$ is the fractional
  vertex cover of $H$, and $C, C' > 0$ are constants that depend only
  on $H$.
\end{theorem}

\noindent
This result requires the minimum vertex cover to assign nonzero weight
to every vertex; this happens for cycles but not stars.

Consider the fractional vertex cover of $H$
\begin{equation}\label{eq:vertex-cover}
\begin{split}
\text{min}&\sum_{a\in V_H} x_a\\
\text{s.t.}&\sum_{a\in e} x_a \geq 1\text{~for all~}e\in E_H,
\end{split}
\end{equation}
let $x^*\in \R^{V_H}$ denote an optimal solution and let $\tau$ denote its value. 

Fix a mapping $\chi: V_G\to V_H$ (see Algorithm~\ref{alg:sample}, line~3).  For a subset $S\subseteq V_G$ we write $\chi(S)\sim H$ if the subgraph induced by $S$ equipped with labels $\chi(S)$ contains a copy of $H$, i.e. for every $a\subseteq S$  one has that if $\chi(a)\in E_H$, then $a\in E_G$. Note that, if $A(H)$ is the number of automorphisms of $H$, the probability that a randomly chosen $\chi$ will give $\chi(S) \sim H$ is $A(H) / k^k$.

\begin{algorithm}
	\caption{Subgraph counting by vertex sampling}\label{alg:sample}
	\begin{algorithmic}[1]
		\Procedure{Sample}{$H, p$}  \Comment{Input: hypergraph $H$, sampling probability $p$}
		\State Compute minimum vertex cover $x^*$ in $H$
		\State $\chi\sim UNIF([k]^{V_G})$ \Comment{Random mapping of $V_G$ to $V_H=[k]$}
		\For {$u\in V_G$}
		\State $a\gets \chi(u)$
		\State $X_u\gets $ independent Bernoulli r.v. with mean $p^{x^*_a}$
		\EndFor
		\State $E'\gets \{e\in E_G: \chi(e)=|e|\text{~and~}\prod_{u\in e} X_u=1\}$\Comment{Keep colorful edges only}
		\State $Z\gets k^k\cdot p^{-\tau}\cdot \sum_{\substack{S\subseteq V_G: \chi(S)\sim H}} \prod_{u\in S} X_u$\Comment{Knowing $E'$ and $\chi$ suffices to compute $Z$ }
		\State \textbf{return} $Z/A(H)$ 
		\EndProcedure
	\end{algorithmic}
\end{algorithm}

\begin{lemma}\label{lm:space}
For every $G=(V_G, E_G)$, every $H=(V_H, E_H)$ with $|V_H|=k$, if $E'$ is the set of edges sampled by Algorithm~\ref{alg:sample} (line~8), then $\expect[|E'|]\leq  p |E_G|$.
\end{lemma}
\begin{proof}
For every choice of $\chi: V_G\to [k]$, only edges $e\in E_G$ with $\chi(e)=|e|$ are kept, and each such edge is kept with probability 
\begin{equation*}
\begin{split}
\expect\left[\prod_{u\in e} X_u\right]=\prod_{u\in e} \expect\left[X_u\right]=\prod_{u\in e} p^{x^*_{\chi(u)}}=p^{\sum_{a\in e} x^*_a}\leq p, 
\end{split}
\end{equation*}
where we used the fact that for every $e\in E_H$ one has $\sum_{a\in e} x^*_a\geq 1$ since  $x^*$ is a feasible vertex cover. Thus, 
the number of edges that the algorithm keeps is at most $p |E_G| $ in expectation. 
\end{proof}

\begin{lemma}\label{lm:mean}
For every $G=(V_G, E_G)$, every $H=(V_H, E_H)$ with $|V_H|=k$ the estimator $Z$ computed by Algorithm~\ref{alg:sample} satisfies $\expect[Z]=A(H)T$.
\end{lemma}
\begin{proof}
We have $Z=k^k p^{-\tau}\sum_{S\subseteq V_G} {\bf I}[\chi(S)\sim H]\cdot \prod_{u\in S} X_u$, so 
\begin{equation*}
\begin{split}
\expect[Z]&= k^k p^{-\tau}\sum_{S\subseteq V_G} \prob[\chi(S)\sim H] \cdot p^{\tau}\\
&=k^k \sum_{S\subseteq V_G} \prob[\chi^{-1} \text{is an isomorphism form $H$ to $S$}]\\
&=A(H)T,\\
\end{split}
\end{equation*}
as required.
\end{proof}

\noindent
The following simple claim will be useful for upper bounding the variance:
\begin{claim}\label{cl:st-counting}
For every hypergraph $H$, every hypergraph $G$ with vertex degrees bounded by $d$ the following conditions hold. For every $S\subseteq V_G$ the number of sets $U\subseteq V_G$ such that $\chi(S)\sim H$, $\chi(U)\sim H$ for some $\chi:V_G\to [k]$ and $|S\cap U|=r, r>0$ is upper bounded by 
$d^r f(H)$ for some function $f$ of the hypergraph $H$.
\end{claim}
\begin{proof}
We bound the number of $U\subseteq V_G$ such that $\chi(U)\sim H$ for some $\chi:V_G\to [k]$ and $S\cup U\neq \emptyset$. Fix one such $U$, and define an auxiliary graph $J=(V_J, E_J)$ on vertex set $(U\setminus S)\cup \{s\}$, where $s$ is a supernode corresponding to $S$, by connecting two vertices $a, b\in J$ by an edge if there exists $\chi|_U:U\to V_H$ and an edge $e$ in $\chi(U)\cap E_H$ that includes both, and give the edge $(a, b)\in E_J$ label $e$. Here we say that an edge includes supernode $s$ if it has nonempty intersection with $S$. Since $U$ is connected, there exists a spanning tree $F$ in the graph $J$ whose edges are labeled by edges of $H$. We call such a spanning tree $F$ a {\em template}, and we say that the pair $(S, U)$ is consistent with template $F$. Given $S$, the number of possible $U$'s consistent with template $F$ is upper bounded by $(d k^2)^{|F|}$. Indeed, starting with $S$, one can traverse the edges of the forest $F$ to discover all vertices in $U\setminus S$, with at most $d$ edges incident on every vertex in $G$ by assumption of the lemma, at most $k$ choices of the next vertex within a given edge, and at most $k$ choices of a vertex in $S$ to start from when starting to traverse a subtree subtended at $s$ in $F$. The number of templates $F$ is a function of the graph $H$ only, and since $|F|\leq r$ (recall that $r=|U\setminus S|$ by definition), we get the bound of $d^r f(H)$ for some function $H$.
\end{proof}

We will need the following, for a proof see e.g.\,~\cite{havet:lecturenotes}
\begin{claim}\label{cl:half-integrality}
For every graph $H=(V_H, E_H)$ the optimal vertex cover $x^*$ can be assumed to be half-integral, i.e. $x^*_a\in \{0, 1/2, 1\}$ for all $a\in V_H$.
\end{claim}

\begin{lemma}\label{lm:var}
If every vertex $v\in V_G$ belongs to at most $d$ hyperedges and $H$ is connected, then one has ${\bf Var}[Z]\leq d^k f(H) k^{k} p^{-\tau}\expect[Z]$.

Furthermore, if $H=(V_H, E_H)$ is a connected graph (i.e. every hyperedge has size $2$) that has an optimal vertex cover with full support (i.e. one that does not assign zero weight to any vertex), then for every graph $G=(V_G, E_G)$ with degrees bounded by $d\leq \frac1{2}p^{1/2}$ one has ${\bf Var}[Z]\leq 2f(H) (d p^{1/2})  p^{-\tau} \expect[Z]$.
\end{lemma}
\begin{proof}
Let $k$ denote the number of vertices in $H$.

We have 
\begin{equation*}
\begin{split}
Z^2&=\left(k^k p^{-\tau}\sum_{S\subseteq V_G} {\bf I}[\chi(S)\sim H]\cdot \prod_{u\in S} X_u\right)^2\\
&=k^{2k} p^{-2\tau}\sum_{S, U\subseteq V_G} {\bf I}[\chi(S)\sim H\text{~and~}\chi(U)\sim H]\cdot \prod_{u\in S\cup U} X_u.\\
\end{split}
\end{equation*}

Taking expectations over $\chi$ and $X$, we get
\begin{equation}\label{eq:variance-bound2}
\begin{split}
\expect[Z^2]&=k^{2k} p^{-2\tau}\sum_{S, U\subseteq V_G} \prob[\chi(S)\sim H\text{~and~}\chi(U)\sim H]\cdot \prob[X_u=1\text{~for all~}u\in S\cup U]\\
&=k^{2k} p^{-2\tau}\sum_{\substack{S, U\subseteq V_G\\S\cup U=\emptyset}} \prob[\chi(S)\sim H\text{~and~}\chi(U)\sim H]\cdot \prob[X_u=1\text{~for all~}u\in S\cup U]\\
&+k^{2k} p^{-2\tau}\sum_{\substack{S, U\subseteq V_G\\S\cup U\neq \emptyset}} \prob[\chi(S)\sim H\text{~and~}\chi(U)\sim H]\cdot \prob[X_u=1\text{~for all~}u\in S\cup U]\\
&=\expect[Z]^2+Q,
\end{split}
\end{equation}
where
\begin{equation}\label{eq:234hth}
\begin{split}
Q&=k^{2k} p^{-2\tau}\sum_{\substack{S, U\subseteq V_G\\S\cup U\neq \emptyset}} \prob[\chi(S)\sim H\text{~and~}\chi(U)\sim H]\cdot \prob[X_u=1\text{~for all~}u\in S\cup U]\\
&\leq k^{2k} p^{-2\tau}\sum_{\substack{S, U\subseteq V_G\\S\cup U\neq \emptyset}} \prob[\chi(S)\sim H]\cdot \prob[X_u=1\text{~for all~}u\in S\cup U]\\
&\leq k^{2k} p^{-2\tau}\sum_{S\subseteq V_G} |\{U\subseteq V_G: U\cap S\neq \emptyset \text{~and~}\exists \chi\text{~s.t.}\chi(U)\sim H\}|\cdot\\
&~~~~~~~~~~~~~~~~~~~~~~~~~\cdot \prob[\chi(S)\sim H]\cdot \prob[X_u=1\text{~for all~}u\in S\cup U]\\
\end{split}
\end{equation}

By Claim~\ref{cl:st-counting} we get that 
\begin{equation}\label{eq:counting-ub}
|\{U\subseteq V_G: U\cap S\neq \emptyset \text{~and~}\exists \chi\text{~s.t.}\chi(U)\sim H\}|\leq d^{|U\setminus S|} f(H)
\end{equation}
for some function $H$,  substituting this bound into~\eqref{eq:234hth} and using the upper bound $\prob[X_u=1\text{~for all~}u\in S\cup U]\leq \prob[X_u=1\text{~for all~}u\in S]$ as well as $|U\setminus S|\leq k$, we get 
\begin{equation*}
\begin{split}
Q&\leq  d^k f(H) k^{2k} p^{-2\tau} \sum_{\substack{S\subseteq V_G}} \prob[\chi(S)\sim H]\cdot \prob[X_u=1\text{~for all~}u\in S]\\
&=  d^k f(H) k^{k} p^{-\tau} \expect[Z].
\end{split}
\end{equation*}
Putting this together with~\eqref{eq:variance-bound2} and using ${\bf Var}(Z)=\expect[Z^2]-\expect[Z]^2$, we get
$$
{\bf Var}[Z]=d^k f(H) k^{k} p^{-\tau} \expect[Z],
$$
proving the first claim of the lemma.

For the second claim of the lemma first note that by the half-integrality of vertex cover for graphs (Claim~\ref{cl:half-integrality}) as well as the assumption that the vertex cover of $H$ has full support we have 
\begin{equation}\label{eq:0hg9gjeg}
\prob[X_u=1\text{~for all~}u\in S\cup U]\leq \prob[X_u=1\text{~for all~}u\in S]\cdot p^{|U\setminus S|/2}.
\end{equation}

We now get, using~\eqref{eq:234hth}, that
\begin{equation*}
\begin{split}
Q&\leq k^{2k} p^{-2\tau}\sum_{S\subseteq V_G} |\{U\subseteq V_G: U\cap S\neq \emptyset \text{~and~}\exists \chi\text{~s.t.}\chi(U)\sim H\}|\cdot\\
&~~~~~~~~~~~~~~~~~~~~~~~~~\cdot \prob[\chi(S)\sim H]\cdot \prob[X_u=1\text{~for all~}u\in S\cup U]\\
&= k^{2k} p^{-2\tau}\sum_{S\subseteq V_G} \sum_{r\geq 1} |\{U\subseteq V_G: |U\cap S|=r \text{~and~}\exists \chi\text{~s.t.}\chi(U)\sim H\}|\cdot\\
&~~~~~~~~~~~~~~~~~~~~~~~~~\cdot \prob[\chi(S)\sim H]\cdot \prob[X_u=1\text{~for all~}u\in S\cup U]\\
&\leq  k^{2k} p^{-2\tau}\sum_{S\subseteq V_G} \sum_{r\geq 1} d^r f(H) \cdot \prob[\chi(S)\sim H]\cdot \prob[X_u=1\text{~for all~}u\in S\cup U]\\
&\leq  f(H) k^{2k} p^{-2\tau}\sum_{S\subseteq V_G} \sum_{r\geq 1} d^r p^{r/2} \cdot \prob[\chi(S)\sim H]\cdot \prob[X_u=1\text{~for all~}u\in S]\\
&=  f(H) k^{2k} p^{-2\tau}\left(\sum_{r\geq 1} d^r p^{r/2}\right)\sum_{S\subseteq V_G} \prob[\chi(S)\sim H]\cdot \prob[X_u=1\text{~for all~}u\in S]\\
&\leq  2f(H) (d p^{1/2}) k^{2k} p^{-2\tau}\sum_{S\subseteq V_G} \prob[\chi(S)\sim H]\cdot \prob[X_u=1\text{~for all~}u\in S].\\
&=  2f(H) (d p^{1/2}) k^{k} p^{-\tau}\expect[Z].\\
\end{split}
\end{equation*}
In the equation above the second inequality is by~\eqref{eq:counting-ub}, the third is by~\eqref{eq:0hg9gjeg}, and the last is by summing the geometric series, which is justified  due to the assumption $d \leq \frac1{2}p^{1/2}$ of the lemma.

Putting the bounds above together, we get 
$$
{\bf Var}[Z]\leq 2f(H) (d p^{1/2})  p^{-\tau} \expect[Z],
$$
as required.

\end{proof}

We now give 

\begin{proofof}{Theorem~\ref{thm:hypergraphs-ub}}
Let $s=(100 d^k f(H))^{1/\tau} \cdot \e^{-2/\tau} \cdot m \cdot (T/A(H))^{-1/\tau}$, where $m=|E_G|$ is the number of edges in $G$.

We consider two cases. If $s>m$, then we simply sample all the edges of $G$ and compute the number of copies of $H$ offline. If $s<m$, we use Algorithm~\ref{alg:sample} with the sampling parameter $p$ set to $p=s/m$. Note that by Lemma~\ref{lm:space} the space complexity is at most $s$. We get by Lemma~\ref{lm:mean} that our estimator is unbiased, and by Lemma~\ref{lm:var} (first part) that its variance is 
${\bf Var}[Z]\leq d^k f(H) (m/s)^{\tau} \expect[Z]=d^k f(H) (m/s)^{\tau} \expect[Z]$.

Since $s=(100 d^k f(H))^{1/\tau} \cdot \e^{-2/\tau} \cdot m T^{-1/\tau}$ by our setting above, we get that 
\begin{equation*}
\begin{split}
{\bf Var}[Z]&\leq d^k f(H) (m/s)^{\tau} \expect[Z]\leq  \e^2 T \expect[Z]=  \frac1{100} \e^2 T^2,
\end{split}
\end{equation*}
and the theorem follows by Chebyshev's inequality. 
\end{proofof}

\begin{proofof}{Theorem~\ref{thm:graphs-ub}}
Let $s=(100 d^k f(H))^{1/\tau} \cdot \e^{-2/\tau} \cdot m \cdot (T/A(H))^{-1/\tau}$, where $m=|E_G|$ is the number of edges in $G$.

We consider two cases. If $s>m$, then we simply sample all the edges of $G$ and compute the number of copies of $H$ offline. If $s<m$, we use Algorithm~\ref{alg:sample} with the sampling parameter $p$ set to $p=s/m$. Note that by Lemma~\ref{lm:space} the space complexity is at most $s$. 
We get by Lemma~\ref{lm:mean} that our estimator is unbiased, and by Lemma~\ref{lm:var} (second part) that its variance is 
\begin{equation*}
\begin{split}
{\bf Var}[Z]&\leq 2f(H) (d p^{1/2}) \cdot (m/s)^{\tau} \expect[Z]\\
&\leq 2f(H) (d (s/m)^{1/2}) \cdot (m/s)^{\tau} \expect[Z]\\
&\leq 2f(H) \cdot (m/s)^{\tau} \expect[Z]
\end{split}
\end{equation*}
 under the assumption that $d\leq (m/s)^{1/2}$ (we verify this assumption shortly). %$d\leq \e^{1/\tau} T^{1/(2\tau)}$.

Since $s=  (100 f(H))^{1/\tau} \cdot \e^{-2/\tau} \cdot m T^{-1/\tau}$ by assumption of the theorem, we get that 
\begin{equation*}
\begin{split}
{\bf Var}[Z]&\leq f(H) (m/s)^{\tau} \expect[Z]\leq  \e^2 T \expect[Z]=  \frac1{100} \e^2 T^2,
\end{split}
\end{equation*}
and the theorem follows by Chebyshev's inequality.  It remains to verify that $d\leq (m/s)^{1/2}$ under our setting of $s$, which is indeed true since 
$$
(m/s)^{1/2}=\left( (100 f(H))^{-1/\tau} \cdot \e^{2/\tau} \cdot T^{1/\tau}\right)^{1/2}\geq \frac1{10} f(H)^{-1/2}  \e^{1/\tau} T^{1/(2\tau)}\geq d
$$
by assumption of the theorem, as required.
\end{proofof}

\bibliographystyle{alpha}
\bibliography{refs}
\begin{appendix}
%!TEX root = ./main.tex

\section{Proofs Omitted from Section~\ref{sec:prelims}}\label{app:a}
\begin{proofof}{Lemma~\ref{lem:fdecomp}}
We will write $f'_i : \lbrace 0, 1\rbrace^{kn} \rightarrow \mathbb{R}$ for the function given by: \[
f'_i((z_i)_{i = 1}^{k} ) = f_i(z_i)
\]
We first note that, if $s_j \not= 0$ for any $j \not= i$, $\widehat{f'_i}((s_j)_{j = 1}^k) = 0$. To show this, let $j, l$ be such that $(s_j)_l = 1$. Then partition the elements of $\lbrace 0, 1\rbrace^{kn}$ into pairs $z, z'$ where $z'$ is obtained by flipping $(z_j)_l$. Then $f'_i(z) = f'_i(z')$ while $\chi_{(s_j)_{j =1}^k}(z) = -\chi_{(s_j)_{j = 1}^k}(z')$, and so $f'_i(z)\chi_{(s_j)_{j = 1}^k}(z) + f'_i(z')\chi_{(s_j)_{j = 1}^k}(z') = 0$. Therefore: 
\begin{align*}
\widehat{f_i}((s_j)_{j =1}^k) &= \sum_{z \in \lbrace 0, 1\rbrace^{kn}} f_i(z)\chi_{(s_j)_{j =1}^k}(z)\\
&= 0
\end{align*}

Now, if $(s_j)_{j =1}^k$ has $s_j = 0$ for all $j \not = i$:
\begin{align*}
\widehat{f'_i}((s_j)_{j = 1}^{k}) &= \frac{1}{2^{kn}} \sum_{(z_j)_{j = 1}^{k} \in \lbrace 0, 1\rbrace^{kn}} f_i'((z_j)_{j = 1}^{k} ) (-1)^{(z_j)_{j = 1}^{k}  \cdot (s_j)_{j = 1}^{k}}\\
&= \frac{1}{2^{kn}} \sum_{(z_j)_{j = 1}^{k}  \in \lbrace 0, 1\rbrace^{kn}} f_i(z_i)(-1)^{z_i \cdot s_i}\\
&= \frac{2^{(k-1)n}}{2^{kn}} \sum_{z \in \lbrace 0, 1\rbrace^{n}} f_i(z)(-1)^{z \cdot s_i}\\
&= \widehat{f_i}(s_i)
\end{align*}

So then, as \[
f = \prod_{i = 1}^k f'_i
\]
we can apply the convolution theorem for Fourier transforms:
\[
\widehat{f}((s_i)_{i = 1}^k) = \sum_{t^{(2)} \in \lbrace 0, 1\rbrace^{kn}, \dots, t^{(k)} \in \lbrace 0, 1\rbrace^{kn}}\widehat{f'_{1}}\left((s_i)_{i = 1}^k \oplus \bigoplus_{i = 2}^{k}t^{(i)}\right) \prod_{i=2}^{k} \widehat{f'_i}\left(t^{(i)}\right).
\]
Now, in the above sum, for each $i = 2, \dots, |E|$, $\widehat{f_i}(t^{(i)})$ will be zero if $t^{(i)}$ has any ones outside of $(t^{(i)})_{i}$. So the only non-zero term of this sum is the one where $t^{(i)} = (0, \dots, s_i, \dots, 0)$ for $i = 2, \dots, k$. Therefore:
\begin{align*}
\widehat{f}((s_i)_{i = 1}^k) &= \prod_{i=1}^k\widehat{f'_i}((0, \dots, s_i, \dots, 0))\\
&= \prod_{e \in E}\widehat{f_i}(s_i).
\end{align*}
\end{proofof}

\begin{proofof}{Lemma~\ref{lem:kklcor}}
We apply the KKL lemma with $\delta = \frac{1}{\lambda c}k \in \lbrack 0,1\rbrack$, getting:
\begin{align*}
\frac{2^{2n}}{|A|^2}\sum_{s \in \lbrace 0,1\rbrace^{n}; |s| = k} \widehat{f}(s)^2 &\le \frac{2^{2n}}{|A|^2} \frac{1}{\delta^k}\left(\frac{|A|}{2^{n}}\right)^{\frac{2}{1 + \delta}}\\
& = \frac{1}{\delta^k}\left(\frac{2^{n}}{|A|}\right)^\frac{2\delta}{1 + \delta}\\
&\le \frac{1}{\delta^k}\left(\frac{2^{n}}{|A|}\right)^{2\delta}\\
&\le \frac{2^{2\delta c}}{\delta^k}\\
&= \left(\frac{2^{1/\lambda}\lambda c}{k}\right)^k\\
&\leq \left(\frac{2\lambda c}{k}\right)^k
\end{align*}
\if 0 Then, for the second:
\[
\frac{2^{2n}}{|A|^2}\sum_{s \in \lbrace 0,1\rbrace^{n}; |s| = n - k} \widehat{f}(s)^2 = \frac{2^{2n}}{|A|^2}\sum_{s \in \lbrace 0,1\rbrace^{n}; |s| = k} \widehat{f}(\overline{s})^2
\]
Now, for any $s$:
\begin{align*}
\widehat{f}(\overline{s}) & = \frac{1}{2^{n}}\sum_{z\in\lbrace 0,1\rbrace^{n}} (-1)^{z\cdot \overline{s}}f(z)\\
&= \frac{1}{2^{n}}\sum_{z\in\lbrace 0,1\rbrace^{n}} (-1)^{||z||_1 + z\cdot s}f(z)\\
&= \widehat{f'}(s)
\end{align*}
Where $f'(z) = (-1)^{||z||_1}f(z)$. Now as $f'$ still has range $\lbrace -1, 0, 1\rbrace$ and the same set of non-zeros as $f$, it meets the requirements for the KKL lemma and so by the same calculation as in the first case:

\begin{align*}
\frac{2^{2n}}{|A|^2}\sum_{s \in \lbrace 0,1\rbrace^{n}; |s| = n - k} \widehat{f}(s)^2 &= \frac{2^{2n}}{|A|^2}\sum_{s \in \lbrace 0,1\rbrace^{n}; |s| = k} \widehat{f'}(s)^2\\
&\le \left(\frac{2 \lambda c}{k}\right)^k.
\end{align*}\fi
\end{proofof}

\end{appendix}
\end{document}